%% file: main_manuscript.tex
\documentclass[aps,twocolumn,floatfix,superscriptaddress,nofootinbib]{revtex4-2}
\usepackage{booktabs}
\usepackage{graphicx}
\usepackage{braket}
\usepackage{stix}
\usepackage{xcolor}
\usepackage{appendix}
\usepackage{dcolumn}
\usepackage{soul}
\usepackage{bm}
\usepackage{amsthm}
\usepackage{enumitem}
\usepackage{amsmath}
\newtheorem{theorem}{Theorem}[section]

\newtheorem{lemma}{Lemma}
\usepackage{xr-hyper}
\usepackage{hyperref}
\hypersetup{colorlinks,
	citecolor=blue
}
\usepackage{cleveref}
\usepackage{apptools}
\AtAppendix{\counterwithin{lemma}{section}}

\usepackage[ruled,vlined]{algorithm2e}

\begin{document}

\preprint{APS/123-QED}

\title{Random projection using random quantum circuits}

\author{Keerthi Kumaran}

\affiliation{Department of Physics and Astronomy, Purdue University, West Lafayette, IN 47907}

\author{Manas Sajjan}

\affiliation{Department of Chemistry, Purdue University, West Lafayette, IN 47907}

\author{Sangchul Oh}

\affiliation{Department of Chemistry, Purdue University, West Lafayette, IN 47907}
\affiliation{Department of Physics, Southern Illinois University, Carbondale, Illinois 62901-4401}

\author{Sabre Kais}
\email{kais@purdue.edu}

\affiliation{Department of Physics and Astronomy, Purdue University, West Lafayette, IN 47907}

\affiliation{Department of Chemistry, Purdue University, West Lafayette, IN 47907}

\affiliation{Department of Electrical and Computer Engineering, Purdue University, West Lafayette, IN 47907}

\begin{abstract}
The random sampling task performed by Google's Sycamore processor gave us a glimpse of the "Quantum Supremacy era". This has definitely shed some spotlight on the power of random quantum circuits in this abstract task of sampling outputs from the (pseudo-) random circuits. In this manuscript, we explore a practical near-term use of local random quantum circuits in dimensional reduction of large low-rank data sets. We make use of the well-studied dimensionality reduction technique called the random projection method. This method has been extensively used in various applications such as image processing, logistic regression, entropy computation of low-rank matrices, etc. We prove that the matrix representations of local random quantum circuits with sufficiently shorter depths ($\sim O(n)$)
serve as good candidates for random projection. We demonstrate numerically that their projection abilities are not far off from the computationally expensive classical principal components analysis on MNIST and CIFAR-100 image data sets. We also benchmark the performance of quantum random projection against the commonly used classical random projection in the tasks of dimensionality reduction of image datasets and computing Von Neumann entropies of large low-rank density matrices. 
And finally using variational quantum singular value decomposition, we demonstrate a near-term implementation of extracting the singular vectors with dominant singular values after quantum random projecting a large low-rank matrix to lower dimensions. All such numerical experiments unequivocally demonstrate the ability of local random circuits to randomize a large Hilbert space at sufficiently shorter depths with robust retention of properties of large datasets in reduced dimensions. 
\end{abstract}

\maketitle


\section{Introduction}

Many problems in machine learning and data science involve the dimensional reduction of large data sets with low ranks \cite{VELLIANGIRI2019104} (e.g. Image processing). Dimensional reduction as a preprocessing step reduces computational complexity in the later stages of processing. Principal Components Analysis (PCA) \cite{SALEM2019292}, reliant on Singular Value Decomposition (SVD), is one such method to reduce the dimension of data sets by retaining only the singular vectors with dominant singular values. There are quantum circuit implementations for PCA (and SVD) \cite{Lloyd_2014,PhysRevA.97.012327,kerenidis2016quantum,Wang2021variationalquantum} and for related applications \cite{PhysRevLett.120.050502}, some of which are near-term (Noisy Intermediate Scale Quantum (NISQ) technologies era \cite{Preskill2018quantumcomputingin}) algorithms \cite{Wang2021variationalquantum}. \\

Techniques like PCA (and SVD) involve a complexity of $O(N^3)$, where $N$ is the size (or the dimension) of data vectors. An alternative to such computationally expensive dimensional reduction methods is the random projection method \cite{10.1145/2842602,ACHLIOPTAS2003671,image}. In the random projection method, we multiply the data sets with certain random matrices and project them to a lower dimensional subspace. Recent years have witnessed fruitful usage of an especially thoughtful variant of such random projections which are known to preserve 
the distance between any two vectors in the data set (say \textbf{$\vec{x_1}$} and \textbf{$\vec{x_2}$})
in the projected subspace up to an error that scales as $O(\sqrt{\frac{\log(N)}{k}})$ where $N$ is the original dimension and $k$ is the reduced dimension of each data vector.  This choice is motivated by the Johnson Lindenstrauss lemma (JL lemma) \cite{JL} introduced at the end of the last century. Since this manuscript will exclusively use such transformations to validate all the key results, we shall hereafter refer to such candidates as good random projectors.  
Such projection techniques are beneficial to myriad applications because the preservation of distances between data vectors ensures that their distinctiveness is uncompromised thereby rendering them usable for discriminative tasks such as classification schemes like logistic regression  (\cite{paul2014random}).

Classically, this is advantageous compared to other methods like PCA  because the random matrix used for projection is independent of the data set considered. The time complexity involved in the random projection arises from matrix multiplication complexity $O(N^{2.37})$ \cite{alman2020refined} followed by the usual SVD complexity of $O(N^2 {\rm {poly}} \log(N))$ 
making the resulting scheme cheaper than the PCA (or SVD). It must be emphasized that a further reduction in the time complexity to $O(N \rm poly\log(N))$ can be afforded using the Fast Johnson-Lindenstrauss transforms \cite{doi:10.1137/060673096}. Several candidates have been studied in classical random projection including Haar random matrices, Gaussian random matrices etc. But the memory complexity of storing such matrices can be potentially huge (proportional to $N^2$ times the precision of each matrix entry). This has engendered the introduction of several competing candidates with better memory complexity(containing sparse matrices with random integer entries) and multiplication complexity. The latter category is mainly considered in practical applications today\cite{ACHLIOPTAS2003671} and will also be used to compare the results of the quantum variants in this manuscript.\\

Classically random projections performed by using projectors sampled from Haar random unitaries suffer from the innate problem of storage due to its exceptionally high memory usage. Even in the quantum setting, implementing Haar random unitaries requires exponential resources as shown in some counting arguments \cite{knill1995approximation}. As a result, it is natural to consider unitary $t-$designs  which only match the Haar measure up to $t$-th moments. 
Quantum implementation of such $t-$ designs, as has been studied in this manuscript, is efficient owing to the fact that local random quantum circuits approach approximate unitary $t-$ designs \cite{PhysRevA.80.012304,dankert2005efficient,4262758} in sufficiently shorter $(O(\log(N) t^{10.5})$ depths \cite{Harrow_2009,Brand_o_2016} (Here, we have assumed that the number of qubits $n$ required to encode a data vector or a wave vector of size $N$ is $\sim \log (N)$). It was shown recently that even shorter depths suffice \cite{Haferkamp2022randomquantum}. The primary workhorse of this manuscript will be based on such quantum circuits which as we shall eventually show not only performs better in accuracy than standard more commonly used classical variants but also require a lesser number of single qubit random rotation gates $O({\rm poly}(\log(N)))$ for implementation.

The flow of the paper is as follows. In Sec.II, we begin with an introduction to the JL Lemma and how that makes the random projection method effective. This is followed by a brief introduction to the Haar measure and approximate Haar unitaries generated from local random quantum circuits. Then, we explicitly prove that the local random quantum circuits which are exact unitary $2-$ designs can satisfy the JL lemma with the same high probability as Haar random matrices thereby making them good random projectors. We then extend the results to approximate unitary $2-$ designs and discuss the bounds on depths to achieve a certain error threshold in the JL lemma
and derive a slightly different probability of the satisfaction of the latter.  We would note that the quantum memory required to store a 2-design or approximate 2-design is $O({\rm poly}(\log N))$ where $N$ is the size of the data vector.  It is worth noting that it has been previously shown in Ref.\cite{sen2018quantum} that approximate unitary $t-$ designs with $t= O(k)$ can be used to satisfy the JL lemma,  thus corroborating our assertions that even they are good candidates for random projection. The exponentially low limit obtained in Ref\cite{sen2018quantum} is better than the limit derived in this paper only for system sizes $N \ge O(10^4)$. For $N \sim O(10^3)$, limits derived in this manuscript for unitary $2-$ designs are tighter.

For numerical quantification of the key assertions, we first use the MNIST, CIFAR-100 image data sets(\cite{deng2012mnist,Krizhevsky09learningmultiple}) and show that the quantum random projection preserves distances post-projection not far off from the computationally expensive algorithms like PCA (along the lines similar to Ref.\cite{image}) and is similar to the classical random projection. This task doesn't require one to know the singular values or the singular vectors explicitly. We compare the performance of quantum random projection with the commonly used classical random projection technique. Instead of benchmarking it against the Haar random matrices generated classically, we make use of classical random projectors whose storage and multiplications are efficient. To this end, we use
Subsampled Randomised Hadamard Transform (SRHT)(\cite{doi:10.1137/060673096}) for different sizes of data sets (1024, 2048 corresponding to 10 and 11 qubits, respectively). As a second instance, we look at a task that requires us to calculate the singular values of large low-rank data matrices and the singular vectors associated with them. In this regard, we perform the computation of entropies of low-rank density matrices by randomly projecting them to reduced subspace (along the lines of (Refs. \cite{kontopoulou2020randomized,dong2022robust}) to get the dominant singular values post-projection. We also demonstrate that one can construct the simplest quantum random projector by performing quantum random projection and extracting the dominant singular values using the variational quantum singular value decomposition (VQSVD) \cite{Wang2021variationalquantum}. Here, random projection to a lower dimension allows us to optimize using a lower dimensional variational ansatz at one end.  The combined effect of the variational nature of the algorithm and the fact that unitary $t$-designs are short-depth establishes good testing grounds for the implementation of this demonstration in near-term devices \cite{Preskill2018quantumcomputingin}
.These demonstrations highlight the ability of local random circuits to efficiently randomize a large Hilbert space(and hence require exponentially lesser parameters to create a random projector) and serve as good random projectors for dimensionality reduction. \\

\section{THEORETICAL BACKGROUND}\label{Theory}

\subsection{Random Projection}

The random projection method is a computationally efficient technique for dimensionality reduction and is useful in many problems in data science, signal processing, machine learning, etc (See, for example,[\cite{ACHLIOPTAS2003671,image}]). The reason behind the effectiveness of the method stems from the Johnson-Lindenstrauss lemma \cite{JL}.\\

\begin{lemma}
    For any $0 < \epsilon < 1$ and $N \in \mathbb{Z}_{+}$. Let us also consider $k \in \mathbb{Z}_{+}$ s.t. \\
    \begin{equation}
        k \sim O(\frac{\log(N)}{\epsilon^2})
    \end{equation}
    Then, for any set of vectors $S=\{\vec{x_i}\}_{i=1}^N$ with  $\vec{x_i} \in \mathbb{R}^d$, $\exists$ $f : \mathbb{R}^d \rightarrow \mathbb{R}^k$ $s.t.$  $\forall\:\: \vec{x_i},\vec{x_j} \in S$
   \begin{equation}
    (1-\epsilon)|\vec{x_1}-\vec{x_2}|_2\leq|f(\vec{x_1}) -f(\vec{x_2})|_2\leq (1+\epsilon)|\vec{x_1}-\vec{x_2}|_2  
    \label{JL eqn}
\end{equation}
where $||_2$ refers to the $l_2$ norm. 
\end{lemma}
\begin{proof}
{See Lemma  Ref.\cite{JL}}
\end{proof}

\subsubsection{Definition 1: random projections vs Good random projections} 
 
Multiplication with Gaussian or Haar random matrices along with a scaling factor followed by projection to a reduced subspace is one function that obeys Eq.\ref{JL eqn} \cite{ghojogh2021johnsonlindenstrauss,lacotte2020optimal}.  
Essentially, it follows from the fact that the expected value of Euclidean distance post-random projection is equal to the Euclidean distance in the original subspace. And the distances post-random projection are not distorted beyond an $\epsilon$ factor with high probability because the variance of the distances post-random projection is sufficiently low.

From now on, we will consider random projections that satisfy the JL lemma in Eq.\ref{JL eqn} to be good random projections. In this regard, the JL lemma says that any set of $N$ points in a high-dimensional Euclidean space (say, $\mathbb{R}^N$)can be embedded into a lower number of dimensions (say, $ k = O(\epsilon^{-2} \log N)$) by a random projection, preserving all the pairwise distances to within a multiplicative factor of $1 \pm \epsilon$. This is also equivalent to preserving all the pairwise inner products (or angles). Formally,

 \begin{equation}
    (1-\epsilon)|\vec{x_1}-\vec{x_2}|_2\leq|\Pi.\vec{x_1} -\Pi.\vec{x_2}|_2\leq (1+\epsilon)|\vec{x_1}-\vec{x_2}|_2, 
    \label{JL_eqn_projector}
\end{equation}

where $||_2$ refers to the $l_2$ norm and $\vec{x_i} \in \mathbb{R}^N \:\:\forall\: i$ and $\Pi$ denote the random projection matrix of size $k\times N$ (or $N \times k$ in which case the random matrix multiplies the data vectors from the right) which obeys Eq.\ref{JL_eqn_projector} and will be called good random projectors from now onwards. \\

Several other candidates which satisfy JL lemma have been considered for random projection in various applications. These random matrices include Subsampled Randomised Hadamard Transform (SRHT) and Input Sparsity Transform (IST). (\cite{ACHLIOPTAS2003671,doi:10.1137/060673096,7055df679c064c5cb60e75f342fc3c1d,10.1145/2488608.2488622}). These random projectors are database friendly because, unlike Gaussian or Haar random matrices whose storage memory cost is proportional to the number of matrix entries and precision, these could be retrieved by matrices that are sparse and have whole number entries.\\

For benchmarking the quantum random projection in our analysis later, we will be using the SRHT (\cite{Tro11:Improved-Analysis}) to compare the performances of random projection using random quantum circuits. We picked SRHT because we want to compare different random matrices that could be efficiently stored and multiplied. In a classical setting, that would be the SRHT  and in a quantum setting, it would be the 2-designs that act as quantum random projectors.  We construct a SRHT random projector as in the Algorithm \ref{algo:DHS}:

\begin{algorithm}[H]

\label{algo:DHS}
\SetAlgoLined
  \textbf{INPUT:} integers $N = 2^n$,$k$ with $k\sim log(N)/\epsilon^2$ \\
  1. Assign a diagonal matrix \textbf{D} $ \in \mathbb{R}^{N\times N}$ whose elements are independent random signs 1,-1.  \\
  2. Assign a matrix $H\in\mathbb{R}^{N \times N}$ to be the normalized Walsh-Hadamard Matrix. \\
  3.Assign a matrix $S\in\mathbb{R}^{N \times k}$ by randomly sampling $k$ columns from the $N \times N$ identity matrix.\\
  4. The Subsampled Randomised Hadamard Transform matrix (which is our CRP) is obtained as $\Pi_{CRP} = \sqrt{\frac{N}{k}} D\cdot H\cdot S$.\\
 \caption{Constructing classical random projector (CRP) $\Pi_{CRP}^{N \times k}$ which satisfies JL lemma.(see Eq.\ref{JL_eqn_projector}) }
\end{algorithm}

\subsection{Approximate Unitary $t$-designs}

In the next section, we will show that the random matrices sampled uniformly from the Haar measure satisfy JL lemma. Though the exact replication of Haar random unitaries is not possible as a quantum circuit because of the fact that they require exponential resources\cite{knill1995approximation}, we will show that to satisfy JL lemma, exact or approximate $t$-designs\cite{4262758}, which matches the Haar measure only until 2nd moment would suffice. We will introduce the definitions related to the approximate $t$-designs in this section and provide theorems on approximate $t$-(or 2) designs satisfying the JL lemma.\\

\subsubsection{Definition 1: Moment Operator}
The $t^{th}$ moment of a superoperator defined with respect to a probability distribution $\mathcal{U}(N)$ defined on the unitary group $\mathbb{U}(N)$ is defined as 
\begin{equation}
    \Phi^{(t)}_{\mathcal{U}(N)} (\cdot) = \int_{U \sim \mathcal{U}(N)} U^{\otimes t} (\cdot)(U^{\dagger})^{\otimes t} d\nu (U),
\end{equation}
where $d\nu (U)$ is the volume element of the probability distribution $\mathcal{U}(N)$.

\subsubsection{Definition 2: Exact Unitary $t$-design }
Let us define $\Delta^{(t)}_{\mathcal{U}(N)}(\cdot)$ \cite{10.1145/2488608.2488622}\cite{Nakaji_2021}\cite{PRXQuantum.3.010313} as
\begin{align}
    \Delta^{(t)}_{\mathcal{U}(N)}(\cdot) &:= \int_{U \sim \mathcal{U}} U^{\otimes t} (\cdot)(U^{\dagger})^{\otimes t} d\nu (U) \nonumber\\ &- \int_{U^\prime \sim \mu_H} U^{\prime \otimes t} (\cdot)(U^{\prime \dagger})^{\otimes t} d\nu (U^\prime)
\end{align} 
where $\mu_H$ refers to the uniform distribution over the Haar measure. Unitaries like $U$ sampled from a  distribution $\mathcal{U}(N)$ are said to form a $t-$ design
$\rm{iff}\:\: \Delta^{(t)} (X) = \mathbb{0}\:\: \forall \:\: X=f(U) \in \mathbb{U}(N)$. This essentially means that the $\mathcal{U}(N)$ mimics the Haar measure up to the $t$-th moment.\\
\subsubsection{Definition 3: $\alpha$ Approximate unitary $t$-designs}
The unitary group $\mathbb{U}(N)$ is said to form an $\alpha$ approximate unitary $t-$ design iff 
\begin{equation}
    || \Delta^{(t)}_{\mathcal{U}(N)} ||_{\diamond}  \leq {\alpha}/N^t
    \label{approx-t}
\end{equation}
where $|| \cdot ||_{\diamond}$ refers to the diamond norm (see for example \cite{Haferkamp2022randomquantum}). Though the $\alpha$ approximate unitary design definition here involves the diamond norm, formulations using other norms exist\cite{low2010pseudorandomness}, and the theorems in the following section generalize for those formulations as well. Local random quantum circuits with lengths $O(\log(N)(\log(N)+\log(1/\alpha)))$ become $\alpha$ approximate 2-designs \cite{Haferkamp2022randomquantum}. \\

\section{Random quantum circuits as Random projectors}

In this section, we show that local random quantum circuits which are approximate unitary 2-designs(or exact unitary 2-designs) are suitable candidates for the random projection (will be called quantum projectors from now on). We show that quantum projectors satisfy Johnson Lindenstrauss's lemma so that their random projection is a $l_2$ subspace embedding with a very high probability of having a very low error. And if one were to compute specific quantities like entropy, one should quantify whether such random matrices produce projected singular values that are closer to the true singular values with higher probability. In a later section, we will discuss how the projection can be done on real quantum computers and how the reduced dimensional vectors and their singular values can be read out from near-term quantum computers. In the following theorems, let us denote the Haar measure distribution as $\mu_H$ and the distribution corresponding to $\alpha$ approximate $t=2$ design as $\mu_{2,\alpha}$. Proofs of the theorems can be found in the Appendix\ref{PROOFS}.\\

\begin{theorem}
\label{Haarrp}
Let $ U \in$ $\mathbb{U}(N)$ be sampled uniformly from the Haar measure($\mu_H$) and let $\vec{x_1}$, $\vec{x_2}$ $\in$ $\mathbb{R}^N$. Then the matrix $\Pi^{k \times N}$ obtained by considering any k rows of $U$ followed by multiplication with $\sqrt{\frac{N}{k}}$ satisfy 
\begin{equation}
     (1-\epsilon)|\vec{x_1}-\vec{x_2}|_2\leq|\Pi \cdot (\vec{x_1} -\vec{x_2})|_2\leq (1+\epsilon)|\vec{x_1}-\vec{x_2}|_2
\end{equation}
with probability greater than $(1- \frac{N-k}{4kN\epsilon^2})$ with $\epsilon \in R_{\ge 0}$ being the error threshold.
\end{theorem}
\begin{proof}
See Appendix \ref{Proof1}
\end{proof}

\begin{theorem}
\label{2designrp}
Let $ U \in$ $\mathbb{U}(N)$ be sampled uniformly from the $\alpha$ approximate unitary 2- design ($\mu_{2,\alpha}$) and let $\vec{x_1}$, $\vec{x_2}$ $\in$ $\mathbb{R}^N$. Then the matrix $\Pi^{k \times N}$ obtained by considering any k rows of $U$ followed by multiplication with $\sqrt{\frac{N}{k}}$ satisfy
\begin{equation}
     (1-\epsilon)|\vec{x_1}-\vec{x_2}|_2\leq|\Pi \cdot (\vec{x_1} -\vec{x_2})|_2\leq (1+\epsilon)|\vec{x_1}-\vec{x_2}|_2
\end{equation}
with a probability greater than $1- (\frac{N-k}{4 kN\epsilon^2} + \frac{\alpha}{4\epsilon^2 k})$ with $\epsilon \in R_{\ge 0}$ being the error threshold.
\end{theorem}
\begin{proof}
See Appendix \ref{Proof2}
\end{proof}

It is worth mentioning that upper bounds on the distortion have been obtained before with exponential scaling (\cite{sen2018quantum}), namely $2^{4}exp(-2^{-4}\epsilon^2 k)$ for Haar measure and $2^{10}exp(-2^{-10}\epsilon^2k) $ for approximate $t-$designs. These limits are better than the limits obtained here only for ($N,k$) > $10^4$. For the cases that are to be explored in the paper, our limits are tighter than the exponential limits.

\begin{algorithm}[H]

\label{algo:QRP}
\SetAlgoLined
  \textbf{INPUT:} integers $N = 2^n$,$k$ with $k\sim log(N)/\epsilon^2$ \\
  1. Choose $\alpha \sim O(\epsilon ^2 k) $ \\
  2. Construct a local random quantum circuit with length $\sim O(n(n+ \ln(1/\alpha)))$  which is an $\alpha$  approximate 2-design  or pick an exact unitary 2-design ansatz. \\
  3.Append a projection operator $P_k$ acting after the local random quantum circuit to form a  $\Pi^{k \times N}$.\\
 \caption{Constructing quantum random Projector (QRP) $\Pi_{QRP}^{N \times k}$ which satisfy JL lemma \ref{Haarrp},\ref{2designrp}}
\end{algorithm}

For the plots in the experiments section of the paper, we use the ansatz used in \cite{McClean_2018} which is assumed to be an exact 2-design ansatz beyond a certain depth Fig.\ref{fig:rqc}. The main text contains the depths at which the ansatz matches the exact 2-design limit($\sim 150$). There are many candidate local random circuit architectures which are $\alpha$ approximate 2-designs \cite{Haferkamp2022randomquantum}. Instead of studying the projection abilities of different local random circuits architecture, the appendix \ref{NON EXACT} contains some experiments where we look at less expensive ansatz and hence is in an approximate unitary 2- design regime by choosing a lower depth ($\sim 50$) (analogous to \cite{PRXQuantum.3.010313}) of the same circuit Fig. \ref{fig:rqc}.

\begin{figure}[htbp]
  \centering
  \includegraphics[width=0.4\textwidth]{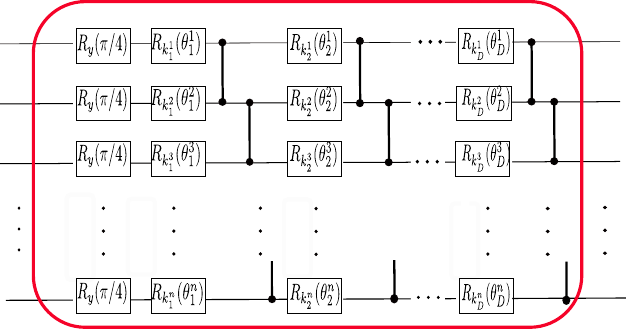}
  \caption{ The local random quantum circuit used in preparing a quantum random projector is the ansatz that has been used in \cite{McClean_2018} and is known to converge to an exact 2-design limit of the variance of the local cost function beyond a certain depth. The circuit contains a layer of $R_y(\pi /4)$ rotations (often used to make all the directions symmetric in a variational training procedure. We don't necessarily need to have this component). This is followed by alternating random single-qubit rotations and ladders of CPHASE operations repeated D times. For $n$ = 10,11 (the dimensions studied in this paper), the circuit reaches the exact 2-design limit (variance limit) at D $\geq$ 150.
}  \label{fig:rqc}
\end{figure}

\section{Experiments on quantum random Projectors}

In this section, we consider two different experiments to benchmark the performance of quantum random projection discussed in the previous section against the SRHT projection which will be labeled as classical random projection in the plots. This should be looked at as a comparison of random projectors' that can be stored and applied efficiently in terms of memory and time complexity in classical vs quantum settings. Since quantum random projectors approximate the Haar measure, their projection abilities are expected to be better than the SRHT projectors because the latter is less random compared to the Haar measure. However, in certain applications, it is known that they both converge to similar performance when the size of the data set tends to infinity \cite{lacotte2020optimal}. We see in the appendix \ref{LARGER} that their performances start becoming closer when we increase the size of the data matrices, and vectors from 1024 to 2048 (corresponding to 10 and 11 qubits, respectively). \\

We initially consider the task that doesn't require us to know the singular values and is concerned with only dimensionality reduction. In this regard, we reduce the dimensions of the MNIST \cite{deng2012mnist} and CIFAR 100\cite{Krizhevsky09learningmultiple} image datasets and benchmark the performance of quantum random projection against classical random projection. We also compare it with the computationally expensive principal components analysis (PCA) which is supposed to give the exact projection to the dominant singular vectors of the datasets and cannot be outperformed beyond a certain rank.\\

In the second task, we calculate the Von Neumann entropy of low-rank density matrices (along the lines of \cite{kontopoulou2020randomized}) which requires us to know the singular values after random projection in addition to the dimensionality reduction. We compare the performance of quantum random projection (QRP) vs classical random projection (CRP) for this task over different ranks (r) of the density matrices.\\

In this section, we pick the local random quantum circuit from Fig. \ref{fig:rqc} and we assume that we can make arbitrary projection operators with any number of basis vectors, i.e, if the $P_k$ operator projects to the first $k$ basis state $(\ket{p_1},\ket{p_2},..,\ket{p_k})$ in any basis. Then, 
\begin{equation}
    P_k = \sum_{i=1} ^k \ket{p_i}\bra{p_i}
\end{equation}
where we don't have any restriction on what basis we pick and what values k can take. In a later section, we discuss the simplest projection operator one can construct by measuring one or more qubits and restricting to particular outputs (0 or 1) in those qubits, as shown in Fig.\ref{fig:QRP}. It is worth noting that this scheme has a structure similar to the quantum autoencoders (\cite{Romero_2017}) but the circuit here is data-agnostic.
\begin{figure}[htbp]
  \centering
  \includegraphics[width=0.48\textwidth]{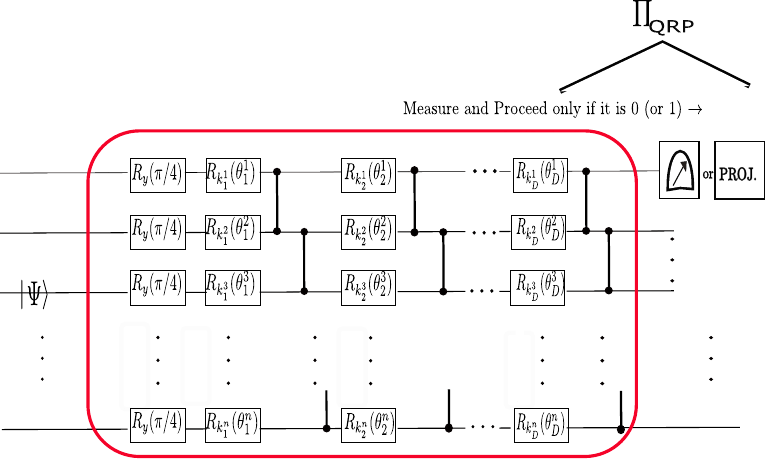}
  \caption{The figure shows the schematic of performing the quantum random projection. The data vector has to be encoded into the circuit through one of the existing encoding schemes {\color{black}(See main text)}. This is followed by the local random quantum circuit  and partial measurements (the number of qubits measured depends on how low your final reduced dimensions are) or an arbitrary projection operator. For partial measurements, the algorithm proceeds only if the measurement results in qubits in only 0 (or only 1). This is equivalent to reducing the data set's size by 1/2,1/4,1/8 and so on depending on how many qubits are measured.
}  \label{fig:QRP}
\end{figure}

\subsection{Dimensionality reduction of Image data sets}
In this subsection, we benchmark the performance of the  QRP against CRP in the task of dimension reduction of subsets of two different image datasets, MNIST and CIFAR-100. We also plot the performance of the computationally expensive PCA which is supposed to capture all the nonzero singular valued singular vectors. When the reduced dimension is greater than the rank of the system, PCA could never be outperformed.\\

MNIST contains 28x28 grayscale images. The matrix representations of the images were boosted to 32x32 so that they can be reshaped into 1024x1 normalised vectors by adding zeros(Note that this is not a common quantum encoding scheme. We use QRP on the normalized data vectors for a direct comparison with CRP). We have to do this preprocessing step because the projectors that we consider (both CRP and QRP) are of the form $2^n \times k$ and hence take only $2^n$ dimensional vectors as the input.\\
 
CIFAR-100, in addition to being 28x28 images, are also colored images and had to be converted to 32x32 grayscale so that they can be fed as input to our projectors. But unlike MNIST, which contains handwritten integers from 0 to 9, the CIFAR-100 dataset contains images belonging to 100 different classes including Airplanes, Automobiles, birds, cats, trucks, etc. As a result, CIFAR-100 is expected to have more features in their datasets and hence greater rank compared to MNIST if we consider subsets from each of these datasets.\\
 
 To perform the comparison between CRP and QRP, we took 1000 images from each of these datasets. And in each of these subsets, we reshaped the images into 1024x1 normalized vectors and performed random projection to lower dimensions (x-axis of the Fig. \ref{fig: JL errors}). Then, we randomly sampled two data vectors and compared the error percentage in their $l_2$ norm (Euclidean distance) between them in the original space and the reduced dimensional space obtained after random projection. This procedure is repeated 10,000 times and the mean error percentages and their 95 $\%$ confidence intervals for different reduced dimensions has been reported in the plots Fig.\ref{fig: JL errors}.\\

 The random projections are performed by multiplying the vectors with random matrices (see Algorithm\:\:\ref{algo:DHS} and Algorithm\:\:\ref{algo:QRP})
 \begin{eqnarray}
    \vec{\tilde {x}} = \Pi_{CRP} \cdot \vec{x} \\
     \ket{\tilde {x}} = P_k.U_{QRP} \cdot \ket{x}
 \end{eqnarray}
where $\Pi_{CRP}$ is the SRHT projector and $U_{QRP}$, $P_k$ are the sampled from the matrix representation of local random quantum circuit and projectors used. And PCA projection is obtained by first computing singular value decomposition on the dataset and projecting them to the subspace of dominant singular vectors.

\begin{equation}
   \vec{\tilde {x}} = \Pi_{PCA} \cdot \vec{x}
\end{equation}
 
\begin{figure}[htbp]
  \centering
  \includegraphics[width=0.48\textwidth]{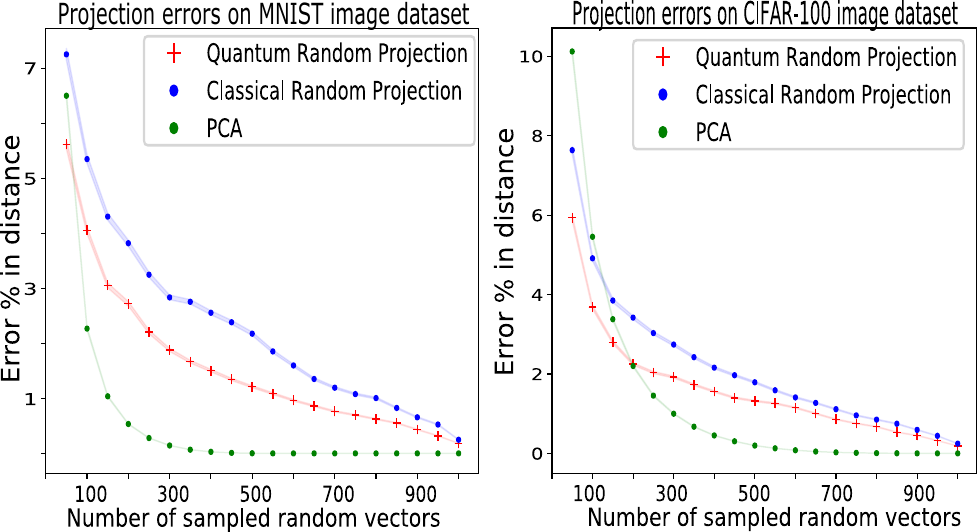}
  \caption{The plots show the mean percentage errors in the distance between 10,000 different random pair of data vectors in the MNIST and CIFAR-100 data sets. The envelopes represent their 95 \% confidence intervals. We see that PCA outperforms the random projection methods beyond a certain rank. Amongst the random projection methods, though there is not much difference between the classical random projection (CRP) and quantum random projection(QRP), we observe that the latter performs slightly better.}
  \label{fig: JL errors}
\end{figure}

Fig.\ref{fig: JL errors} show that the PCA outperforms the random projection methods beyond a certain rank. This is because, beyond the rank of the dataset considered, PCA projects exactly to the subspace with non-zero singular values. Despite that, we see that random projection methods which are not computationally extensive (because they don't compute the subspace with non-zero singular values) perform to the same extent and even better than the PCA at lower reduced dimensions. This dominance in performance at lower reduced dimensions is visible in larger dimensional datasets (see Appendix \ref{LARGER}). We also see that PCA takes more reduced dimensional vectors to catch up with the random projection algorithms in the case of CIFAR-100 because it has comparatively more rank (loosely because it has more features) than the MNIST dataset. {\color{black} These data vectors dimensionally reduced via quantum random projection could be used in quantum machine learning applications such as training an image recognition/classification model (See for example \cite{9550027}).\\\\}

Within the random projection methods, quantum random projection performs slightly better than the classical random projection mainly because Haar random matrices are more random and have tighter JL lemma bounds than the classical random projector. The performance of quantum random projectors which are away from the exact 2-design limit has been analyzed in the Appendix \ref{NON EXACT} by looking at shorter depths ($\sim 50$) of Fig.\ref{fig:rqc} and hence lesser expressive ansatz.\\

{\color{black} The discussion in this section assumed the existence of an exact amplitude encoding scheme for the data vectors. This would require impractical depths of $O(2^n)$ unless the data vectors are genuinely quantum, e.g.,  groundstates of a family of local Hamiltonians. However, for general data vectors like image data vectors, we do not necessarily need exact encoding.  Preserving the distinctness of image data vectors $(\vec{x},\vec{y})$ to a good enough accuracy enables us to use them for many image processing applications such as recognition, and classification. In this regard, there has been substantial work on Approximate amplitude encoding. These schemes encompass approximately encoding data vectors whose amplitudes are all positive \cite{Zoufal_2019}, real\cite{Nakaji_2022}, and even complex data vectors\cite{mitsuda2023approximate} using shallow parametrized quantum circuits.\\\\ With the plots in Fig.\ref{fig: JL errors}, we showed that for exactly encoded data vectors ($\vec{x},\vec{y}$) and quantum random projected vectors ($\vec{\tilde{x}},\vec{\tilde{y}}$)\\
    \begin{equation}
        ||\vec{x}-\vec{y}| - |\vec{\tilde{x}}-\vec{\tilde{y}}|| \leq \delta
        \label{encode1}
    \end{equation}
    on average for pairs of images in the data set used. Here $\delta$ is a very small fraction of $|\vec{x}-\vec{y}|$. A good approximate amplitude encoding scheme is bound to preserve this distance with minimal error, since it preserves the distinctness of the samples as well as (calling $\vec{x}_{app},\vec{y}_{app}$ approximate encoded vectors)\\
    \begin{equation}
        ||\vec{x}-\vec{y}| - |\vec{x}_{app}-\vec{y}_{app}|| \leq \Delta
        \label{encode2}
    \end{equation}
   on average. Here $\Delta$ is a small fraction of $|\vec{x}-\vec{y}|$.\\\\
   With the equations Eqs.\ref{encode1} and \ref{encode2}, it is clear that even with the approximate amplitude encoding, quantum random projection would preserve the distinctness of samples (up to a perturbation of $\delta + \Delta$) and be useful for image processing applications. The exact value of $\Delta$ depends on the efficiency of the approximate encoding used.\\\\
  The other alternative to circumvent the impractical depths of the exact data encoding issue is by adopting different encoding schemes. One can start by reducing the resolution of the images (equivalent to reducing the pixels), which results in reduced classical image data vector dimension (to say $m < 2^n$) and use any other existing data encoding schemes that use qubits' dimensions greater than $m$ but with polynomial depths(For example \cite{Le2011AFR},\cite{2013QuIP...12.2833Z}).\\\\
  If $\Phi(.)$ is the encoding function that takes the original data vector and encodes it as a data vector of dimension $2^d$. Then, to check how well the distinctness is preserved, experiments need to be run on the $d$ qubits with a quantum random projector corresponding to $d$ qubits. Mathematically, we need to check how low the following values are (on average) for two data vectors $\vec{x},\vec{y}$ from the original dataset
    \begin{equation}
        ||\Phi(\vec{x})-\Phi(\vec{y})| - |\Phi(\vec{x})_r-{\Phi(\vec{y})}_r||
        \label{encoding}
    \end{equation}
    where $\Phi(\vec{{x}})_r,\Phi(\vec{y})_r$ are reduced randomly projected encoded vectors.\\\\
    In this work, we confined ourselves to experiments involving an exact encoding scheme despite impractical depths because the preservation of distance for the exact encoding scheme implies the same for the approximate encoding schemes as described earlier. Checking the preservation of distance for other encoding schemes would require knowing the exact form of encoding in Eq.\ref{encoding}.\\\\
     }

Just like its classical counterpart, we can also reconstruct the images back to the original size after the random projection. For classical methods (for a data vector $\vec{x}$ and its reduced data vector $\vec{\tilde{x}}$)

\begin{eqnarray}
    \vec{x}_{recons.} = \Pi_{PCA}^{T} \cdot \vec{\tilde {x}}\\
    \vec{x}_{recons.} = \Pi_{CRP}^{T}\cdot \vec{\tilde {x}} 
\end{eqnarray}
For the quantum case, we need to put in the extra qubits or the subspace to which we projected to get back to the original size and then apply the inverse of the unitary circuit used for projection. For example, if we had projected to the subspace where one of the qubits is in $\ket{0}$ state. We boost the size back to the original size by having a new qubit at $\ket{0}$ and add the inverse unitary circuit ($U_{QRP} ^\dagger$) on this new system. For a general projector,  $\ket{\tilde {x}} \rightarrow $$\ket{\tilde {x}} \otimes \ket{\tilde{p}}$ where tensor product with the $\ket{\tilde{p}}$ ensures that we get back the original size of the data(image). And the reconstruction is done as follows (For the dataset $\ket{x}$)
\begin{equation}
    \ket{x}_{recons.} = U^{\dagger}_{QRP} \cdot \ket{\tilde {x}}
\end{equation}

This is similar to the reconstruction done in \cite{Romero_2017}. These reconstructions work on the premise that the product $\Pi^{\dagger}\Pi \sim {\mathbb{I}}$ of the original data dimension. It is trivial to see that this holds for the PCA projector. It turns out that this also holds for the random projectors. This is because in a larger dimensional space, finding almost orthogonal vectors becomes more common, this has been studied in \cite{HechtNielsen94ContextVectors} and was used in the discussion of \cite{image}. Fig.\ref{fig: recons} shows how one would reconstruct an image from the MNIST dataset after dimensionally reducing a subset of the MNIST images.\\\\
{\color{black} With the reconstructed quantum data vectors, there are processing applications which have computational advantage over classical processing. For example, the complexity of the quantum edge detection algorithm \cite{Yao_2017} is polynomial and doesn't require exponential resources if we have an image encoded either exactly or approximately. The measurement outputs of the edge detection algorithm contain information about the edges. To get the outputs of this experiment, one can also adopt the classical shadows \cite{Huang_2020,Huang_2021} approach to get the probabilities of all the bit strings with less number of measurements than full-state tomography.}

\begin{figure*}[htbp!]
  \centering
  \includegraphics[width=0.8\textwidth]{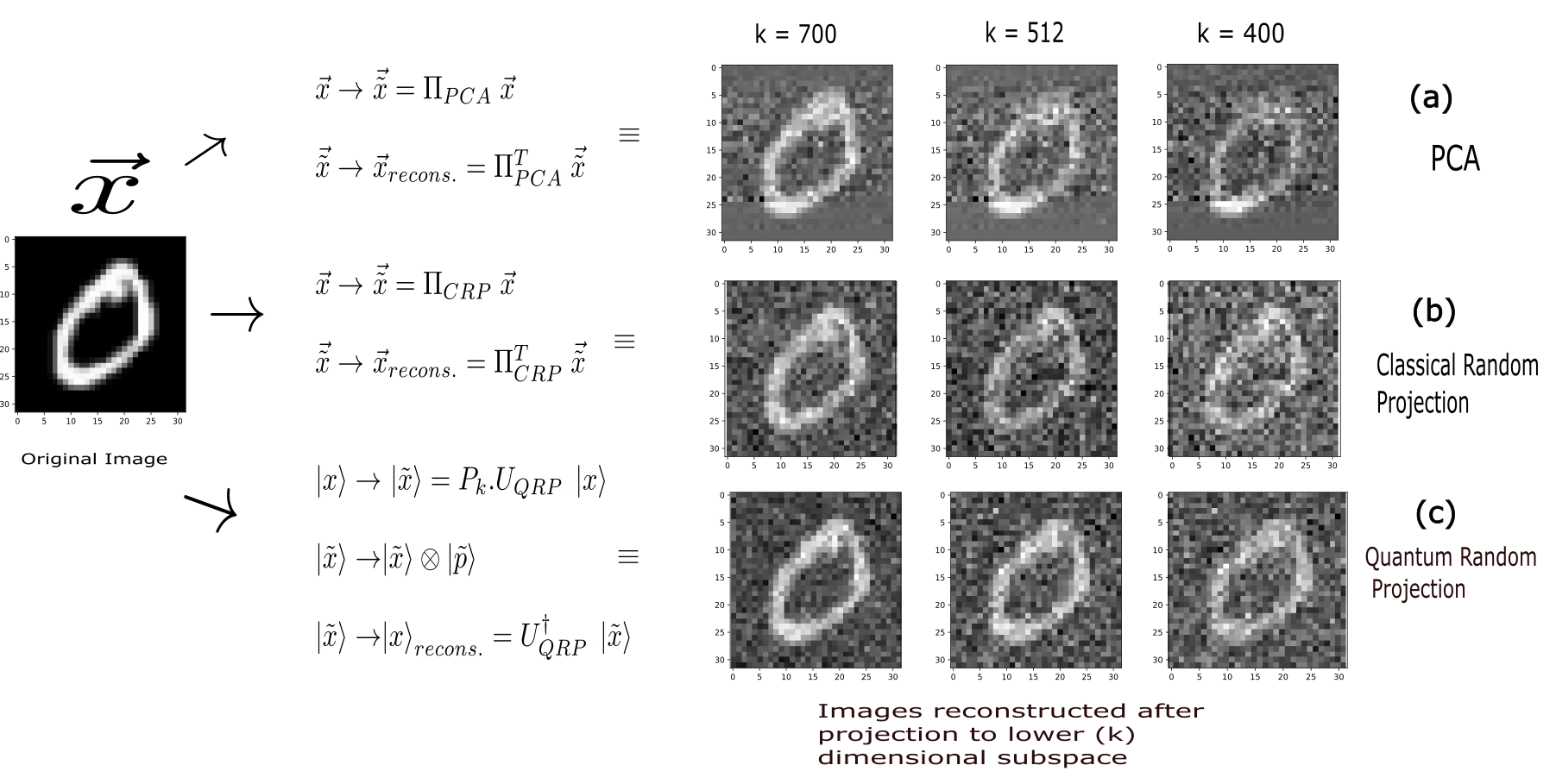}
  \caption{The schematic of steps involved in the dimensionality reduction and image reconstruction of image datasets using (a)PCA, (b)CRP, and (c)QRP. The figure shows the reconstructed images for various reduced dimensions. Though projectors with dimensions 700 and 400 are not straightforward to construct, a reduced dimension of 512 represents projection by measuring one of the qubits (so the size drops from 1024 to 512) and processing further only if it's 0 or 1.{\color{black} The figure illustrates the reconstruction of one of the data vectors from the MNIST data subset we are experimenting, a quantitative description of how the reconstruction performs compared to classical methods has been discussed in the Appendix \ref{imagereconsapp} along with a description about the construction of projection operators.}}
  \label{fig: recons}
\end{figure*}

\subsection{Entropy estimation of low-rank density matrices}

In this subsection, we compare the performance of the quantum random projectors against the classical random projector, SRHT on a task that requires one to obtain the approximate singular values of the dominant singular vectors of a data matrix after the dimensionality reduction. Unlike the previous task, this is concerned with reducing the dimensions of a large data matrix instead of individual data vectors by random projection. After the dimensionality reduction, we check how well the system captures the properties of the dataset by computing the error percentage in a particular property of the data matrix which requires the knowledge of all its singular values.\\

Specifically, we will consider randomly generated semi-positive definite density matrices with random singular vectors but their singular values follow a certain profile. We then compute their entropy after quantum random projection and check their accuracy(along the lines of \cite{kontopoulou2020randomized}). The exact singular values profile of these density matrices depends on the nature of the system. In this experiment, we consider singular values which are linearly decaying and exponentially decaying until the rank of the system and zero afterward. These profiles could be motivated through the existence of physical systems with such profiles. A thermal ensemble of a simple harmonic oscillator mode of frequency $\nu$ with N internal degrees of freedom has an exponentially decaying profile. Here, singular values of $\rho$ will be proportional to $1,e^{-h \nu},e^{-2 h \nu},e^{-3 h \nu},..$ and so on. And, it is known that maximal second order R\'{e}nyi entropy ensemble of a system with a simple harmonic oscillator mode of frequency $\nu$ and N internal degrees of freedom follows a linearly decaying  singular value profile  for its density matrix \cite{Giudice_2021}. The main text contains the plots related to the linearly decaying singular profile and the appendix \ref{EXP} contains the plots related to the exponential decay profile.\\

Here is the procedure to perform random projection given a semi-positive definite matrix $M$ of dimension $N\times N$
\begin{itemize}
    \item Project the original density matrix of size $N\times N$ to a lower dimension $N \times k$ using $\Pi_{CRP}$ and $\Pi_{QRP}$
    \item Perform SVD (classical) or QSVD (quantum) on the lower dimensional matrix to get the singular vectors with singular values $\tilde{p_1}$,$\tilde{p_2}$,$\tilde{p_3}$,....$\tilde{p_k}$ which are approximations to  $p_1,p_2,...p_k$.
    \item Then we obtain an approximation to entropy (S) using $\tilde{S} = \sum \tilde{p_i}\ln{\frac{1}{\tilde{p_i}}}$
\end{itemize}
The accuracy in the approximated entropy is bounded in the following theorem \ref{entropy accuracy}
\begin{theorem}
\label{entropy accuracy}
For a random matrix $\Pi$ satisfying the Johnson-Lindenstrauss (JL) lemma with a distortion $\epsilon \sqrt{\delta}$, where $\epsilon \leq 1/6$ and $\delta \leq 1/2$, the difference in the Von Neumann entropy of a density matrix $\rho$ computed using the random projection with $\Pi$ (denoted as $\tilde{S}(\rho)$) and the true entropy ($S(\rho)$) can be bounded as follows
\begin{equation}
    |{\tilde{S}(\rho) - S(\rho)}| \leq \sqrt{3\epsilon}S(\rho) + \sqrt{\frac{9}{2}}\epsilon
\end{equation}
with probability at least ($1-\delta$).
\end{theorem}

 \begin{proof}{}
See Appendix \ref{Proof3}
\end{proof}

\begin{figure}[htbp]
  \centering
  \includegraphics[width=0.48\textwidth]{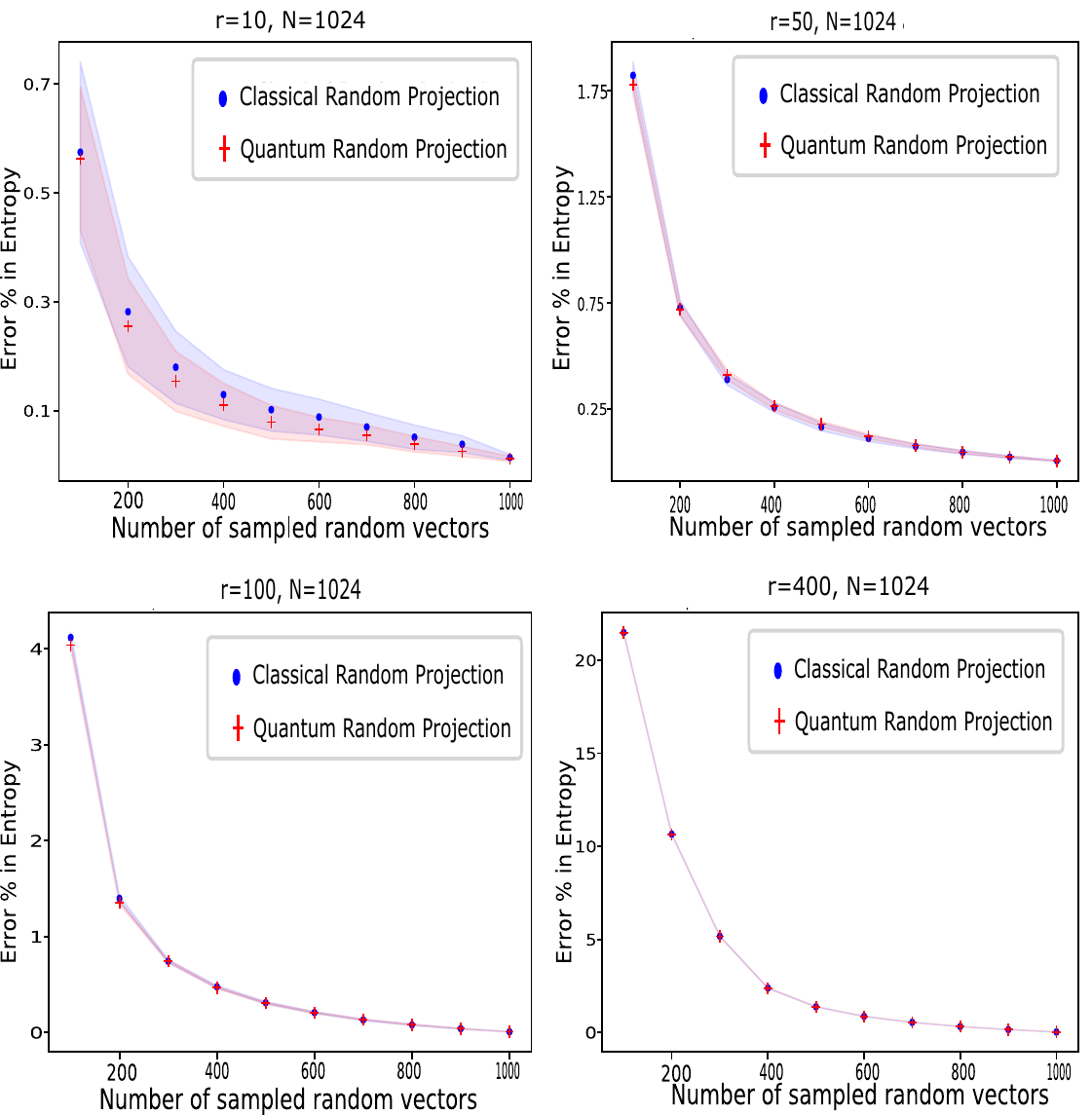}
  \caption{The plots in this figure show the accuracies of quantum random projection and classical random projection in the entropy computation of randomly generated density matrices of size N=1024 and ranks $r=10,50,100,400$ with a linearly decaying singular value profile. The envelopes represent their 90 percent confidence intervals by running the experiments over 100 randomly generated density matrices. The accuracies improve with  decrease in the ranks as expected. }
  \label{fig: linear profile10}
\end{figure}

\begin{figure}[htbp]
  \centering
  \includegraphics[width=0.5\textwidth]{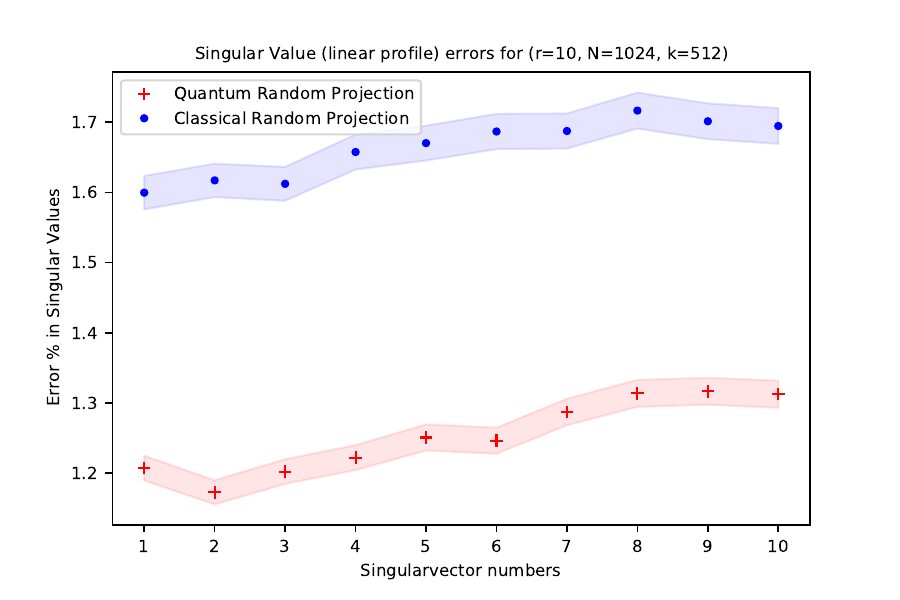}
  \caption{The plot shows the accuracy with which the quantum random projector and the classical random projector pick the singular values of the density matrix for $r = 10$ when reducing the system size by half. The envelope represents 95 \% confidence intervals by running the experiments over 10,000 randomly generated density matrices }
  \label{fig: Individual eigenvalues linear 10}
\end{figure}

The Fig.\ref{fig: linear profile10}  shows the error percentage in the computed entropy after random projection for density matrices of size $1024\times 1024$  with linearly decaying singular values until a certain rank ($r = 10, 50, 100, 40$ in the plot) and zero afterward. The x-axis represents the different reduced dimensions ($k$). The accuracies are better for low ranks as expected and get worse for larger ranks. We observe that the quantum random Projector and the classical random projector perform to similar extents (if not a better quantum performance than classical performance for density matrices 
 with very low rank) in this task. This matches the trends reported in \cite{kontopoulou2020randomized} where they reported similar performance for various other classical random projection matrices like Gaussian, SRHT, and IST. We also show the accuracies with which the quantum and classical random projectors capture the singular values of the system for the rank $r=10$ when the system's size has been reduced by half in figure \ref{fig: Individual eigenvalues linear 10}. Here, we see that the quantum random projectors perform better than their classical counterpart mainly because Haar random matrices that the random circuits try to approximate is more random than any classical random projectors that could be stored with similar or comparable complexity. The appendix contains a discussion regarding how the accuracy improves when we increase the size of the original datasets from 1024 $\times$ 1024 to 2048 $\times$ 2048.\\
 
 We discuss the same plots for density matrices with exponentially decaying singular value profile until a certain rank in the Appendix  \ref{EXP}. We observe there that increasing the rank doesn't change the singular value profile as much and hence the accuracy with which the random projection algorithms work remains pretty much constant. The appendix \ref{NON EXACT} contains the error plots for the accuracy in individual singular values for the case $r=10$ obtained using lesser expressive random circuits (depth $\sim 50$).

\section{How to project in a real Quantum Computer?}

For the quantum random projection to work, in addition to sampling a unitary from the exact (or approximate 2 designs), we also need to have a circuit component for projector operators. In one of the previous sections, we considered arbitrary projection operators which might not be able to be efficiently implemented in a quantum computer with polynomial resources. However, we can look at the simplest projection operations that one can use for the quantum random projection. In Fig.\ref{fig:QRP}, we looked at the simplest projection operation, which is measuring some of the qubits and proceeding only if the qubits are in a certain state ($\ket{0}$ or $\ket{1}$). This is equivalent to projecting it to the subspace where those qubits take that specific value. For example, when you have a circuit of 10 qubits, measuring one of the qubits and proceeding only when that qubit is in $\ket{0}$ means a reduction in the data vector (or ket) dimension from 1024 to 512.\\

However, the projection through measurement discussed above is different compared to the classical projection because a quantum measurement (wavefunction collapse) automatically takes care of the normalization factor and the extra $\sqrt{\frac{N}{k}}$ is not needed. Since the Hilbert space we consider here is large and the ket entries are randomized, the normalization that happens because of the wavefunction is the same as the prefactor we would get in a classical random projection.\\

To demonstrate the quantum random projection with a simple projection operation, we will consider projection operators of the form $\frac{1}{2}(1+\sigma_z^i)$ which is a projection operator to the space where $i^{th}$ qubit is at $\ket{0}$ state.To demonstrate this, we perform such a quantum random projection for a large data matrix with a linearly decaying singular values profile of size $1024\times 1024$ and rank $r =5$ by reducing the data vectors to sizes 512, 256 and 128 by projecting out 1, 2 and 3 qubits respectively. Then, we retrieve the dominant singular vectors by performing a variational quantum singular value decomposition (VQSVD)\cite{Wang2021variationalquantum}. But since the data matrix has been dimensionally reduced, the ansatz we use for finding the right singular vectors is also of reduced size (Fig.\ref{fig: V QSVD}). The details regarding the implementation of VQSVD, and the ansatz type used can be found in the Appendix \ref{VQSVD App}.\\
\begin{figure}[htbp]
  \centering
  \includegraphics[width=0.48\textwidth]{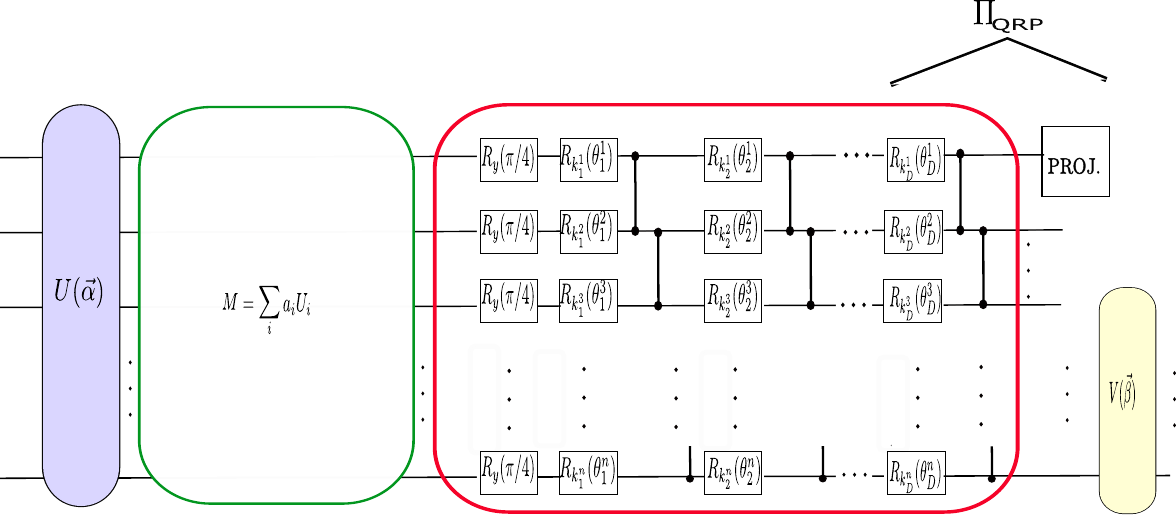}
  \caption{The figure shows the schematic of the variational quantum SVD post quantum random projection to lower dimensions. The data matrix $M$ needs to be loaded using a set of unitary gates with techniques like importance sampling(see related discussion in the appendix of \cite{Wang2021variationalquantum}). Similar to the setup in Fig.\ref{fig:QRP}, we perform projection by measuring a few qubits at the top. This is followed by a training procedure to obtain the dominant singular vectors and their singular values. The singular vectors on the right end belong to the lower dimensional space and hence require a lower dimensional ansatz.}
  \label{fig: V QSVD}
\end{figure}
\begin{figure}[htbp]
  \centering
  \includegraphics[width=0.5\textwidth]{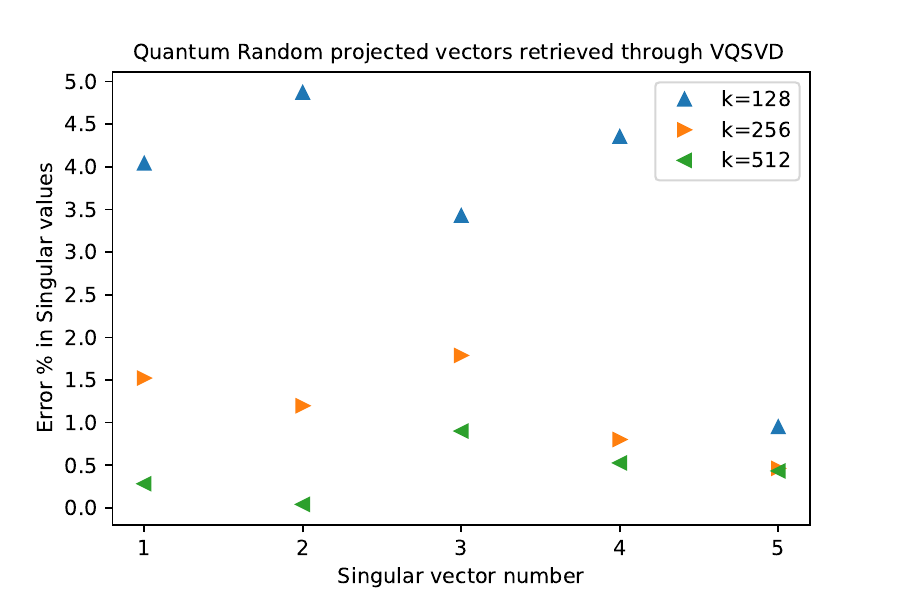}
  \caption{The figure shows the errors in the singular values obtained by reconstructing the dominant singular vectors post quantum random projecting a randomly generated data matrix of rank $r=5$ using variational quantum SVD for various reduced dimensions $k = 512, 256, 128$ .}
  \label{fig: V QSVD results}
\end{figure}

The Fig.\ref{fig: V QSVD results} shows the accuracy with which we were able to reconstruct the singular vectors after quantum random projection for individual singular vectors. This demonstrates how one can perform quantum random projection in near-term devices as the VQSVD algorithm used to retrieve the dominant vectors is a near-term algorithm. The accuracy with which the singular values have been retrieved depends on the expressivity of the ansatz and whether or not it falls into a barren plateau during the training procedure. We haven't discussed the most accurate retrieval of the singular vectors, as that is beyond the scope of this paper. There are many strategies to avoid falling into the barren plateau and improve the convergence rate \cite{Cerezo_2021,Cerezo_2022,Grant2019initialization}. We used the identity block strategy \cite{Grant2019initialization} to avoid barren plateaus (more details on that in the appendix \ref{VQSVD App}.)

\section{Conclusion}\label{concluding_section}
In this work, we explored a practically useful application of local random quantum circuits in the task of dimensional reduction of large low-rank data sets. The main essence of the applicability of the local random circuits in this task is their ability to anticoncentrate rapidly at linear or sub-linear depths \cite{Hangleiter_2018,Dalzell_2022}. This makes them a good random projector to lower dimensions, meaning they preserve the distinctness of different dominant data vectors in a large dataset after dimensional reduction.\\

The theorems discussed in the paper show that just like the Haar random matrices which are good random projectors, their approximate quantum implementations, the exact and approximate t-designs are also good random projectors. The rapid anti-concentration of Hilbert space at linear depths means that the number of random parameters(the random rotation parameters) required to create and reproduce a random projector is logarithmic in the size of the data sets. Such efficiency in the storage complexity of classically generated Haar random matrices or in any classical random projector is not possible. We then, benchmarked its performance against the commonly used classical random projector, Subsampled Randomised Hadamard Transform (SRHT). The quantum random projectors performed slightly better than this classical candidate because they are trying to approximate Haar random matrices which are more random than the classical candidate. We then demonstrated these comparisons for various tasks such as image compression, reconstruction, retrieving the singular values of the dominant singular vectors post-dimension reduction, etc.\\

Though the initial discussion assumed arbitrary projection operators to arbitrary subspaces, we showed simplest projection operators and projection subspaces exist. We demonstrated this simplest quantum random projection and retrieved the dominant singular vectors post-quantum random projection via VQSVD \cite{Wang2021variationalquantum}. This shows the applicability of such quantum random projections and their retrievals in near-term devices. 

{\color{black} Dimensionality reduction facilitated by random projections as discussed in this work can also precede kernel-based variants of PCA wherein eigenvalue decomposition of the Gram matrix associated with the higher-dimensional embedding (often called kernel) is sought \cite{wang2012kernel} especially if the said Gram matrix is low-rank. Beyond the precincts of classical data, such a technique can act as an effective precursor to improve the efficiency of simulation even on quantum data as has been studied in recent work \cite{app13053172}. The essential crux of the idea is heavily rooted on PCA but applied to quantum data wherein repeated Schmidt decomposition of the states and vectorized form of arbitrary operators are performed followed by subsequent removal of singular vectors associated with non-dominant singular values akin to PCA. The techniques explored in this work involving good random projections can be used in conjunction prior to the application of such a protocol to contract the effective space of the states/operators involved. Owing to the demonstrated near-term applicability, similar reduction can also be afforded 
as a preprocessing step in a host of quantum algorithms manipulating quantum data \cite{D2CS00203E} on noisy hardwares. Such protocols are of active interest to the scientific community due to their profound physico-chemical applications ranging from exotic condensed-matter physical systems like Rydberg excitonic arrays\cite{sajjan2022magnetic}, modeling higher dimensional spin-graphical architectures in quantum gravity \cite{borissov1996geometry} and in learning theory of neural networks\cite{sajjan2023imaginary}, constructing unknown hamiltonians through time-series analysis \cite{gupta2023hamiltonian, yu2023robust}, tomographic estimation of quantum state \cite{gupta2022variational,gupta2021maximal},  in electronic structure of molecules and periodic materials \cite{sajjan2021quantum}, quantum preparation of low energy states of desired symmetry \cite{sajjan2021quantum, selvarajan2022variational} or even order-disorder transitions in conventional Ising spin glass using quantum annealers \cite{PhysRevE.105.035305} and quantum variants of Sherrington-Kirkpatrick model \cite{schindler2022variational} to name a few.}\\

We did an extensive comparison using deep ($\sim 150$) exact 2-design ansatz and deferred the discussion about circuits away from the exact 2-design limit to the appendix \ref{NON EXACT}. This is because there exist various random circuit architectures which anti concentrate just like the exact 2-design ansatz and hence could be good candidates for random projection. This could be a good starting point for future study. Also, it is worth studying and constructing quantum random projectors suited for specific applications and datasets( for example, the datasets in health care \cite{Mirniaharikandehei2020ApplyingAR,Heidari2020ApplyingAR}). It has to be noted that the results derived in the main text assumed noiseless quantum gates and measurements. Similar theorems need to be understood for real quantum computers where different noise sources are unavoidable. This leads to a possible future study to understand the extent to which the theorems in the main text are valid on real quantum computers by performing statistical analysis on the bitstrings from the output of real quantum computers (See Refs. \cite{PhysRevA.107.022610,Oh_2022}).

\section{Code availability}
The classical and the quantum random projection matrix(and the rotation parameters used to generate the circuit) used for the comparisons, the data matrix used to generate Fig. \ref{fig: V QSVD results} along with the code for generating the plots in this paper will be made available upon reasonable request. The simulation for the retrieval of dominant singular vectors through VQSVD was done in the Paddle quantum framework \cite{Paddlequantum}.

\section{ACKNOWLEDGEMENTS}

We would like to acknowledge funding from the Office of Science through the Quantum Science Center (QSC), a National Quantum Information Science Research Center,  the U.S. Department of Energy (DOE) (Office of Basic Energy Sciences) under award no. DE-SC0019215, and the National Science Foundation under award number 1955907.

\input{main_manuscript.bbl}

\widetext
\pagebreak
\appendix
\AtAppendix{\counterwithin{lemma}{section}}



\section{Proofs of theorems in the main text}

\label{PROOFS}
\label{Proof1}
\begin{theorem}
\label{Haartheorem}(Theorem \ref{Haarrp} in main manuscript)
Let $ U \in$ $\mathbb{U}(N)$ be sampled uniformly from the Haar measure($\mu_H$) and let $\vec{x_1}$, $\vec{x_2}$ $\in$ $\mathbb{R}^N$. Then the matrix $\Pi^{k\times N}$ obtained by considering any k rows of $U$ followed by multiplication with $\sqrt{\frac{N}{k}}$ satisfy 
\begin{equation}
     (1-\epsilon)|\vec{x_1}-\vec{x_2}|_2\leq|\Pi \cdot (\vec{x_1} -\vec{x_2})|_2\leq (1+\epsilon)|\vec{x_1}-\vec{x_2}|_2
\end{equation}
with probability greater than $(1- \frac{N-k}{4kN\epsilon^2})$ with $\epsilon \in R_{\ge 0}$ being the error threshold.
\end{theorem}

\begin{proof}
Let $P_k$ be the projector operator to any k basis states of the Hilbert space. Without loss of generality, let us consider $\vec{x_1}-\vec{x_2} = \vec{v}$ to be of unit norm. The components of such $\vec{x_1} - \vec{x_2}$ will be $v_1,v_2,v_3,...,v_N$ such that $\Sigma_i |v_i|^2 = 1$.Then,\\ 
\begin{equation}
    |\Pi \cdot \vec{v}|_2^2 = \frac{N}{k} \sum_{i=1}^k |y_{i}|^2
\end{equation}
where $y_1,y_2,...,y_k$ are k components of the vector $U.(\vec{x_1}-\vec{x_2})$ depending on the $P_k$'s projection subspace. Then, \begin{equation}
   E_{U\sim \mu_H}[|y_i|^2] = E_{U\sim \mu_H}[|y_j|^2] 
   \label{Mean}
\end{equation} $\forall i,j$ because we know that the Haar measure satisfies the translational invariance
\begin{equation}
    \int_{U \sim \mu_H} f(U) d\mu(U) = \int_{U \sim \mu_H} f(V U) d\mu (U)
    \label{Haar}
\end{equation} when V is an element belonging to the $U(N)$ group. By choosing $V$ to be the matrix that switches $i$-th and $j$-th components of an $N$-dimensional complex vector, we get the required relation Eq. \ref{Mean}. Now using the fact that $\vec{v}$ is unit norm, we get 
\begin{equation}
   E_{U\sim \mu_H}[|\Pi \cdot \vec{v}|^2_2] = 1
   \label{1st order}
\end{equation}
Now, to show that sampling $U$ from the Haar measure produces a low distortion ($\epsilon$) with a very high probability, we should look at the variance of the projected norm. We should essentially show that $ P_{U \sim \mu_H }((1- 2\epsilon)| \leq |\Pi \cdot \vec{v}|^2_2 \leq (1+ 2\epsilon))$ with a very high probability. Consider Variance $Var_{\mu_H}$
\begin{eqnarray}
Var_{\mu_H} =E_{U\sim \mu_H}[(\sum_{i=1}^k|y_{i}|^2)^2] - E^2_{U\sim \mu_H}(\sum_{i=1}^k|y_{i}|^2) \\
= E_{U\sim \mu_H}[\sum_{i=1}^k|y_{i}|^4] + E_{U\sim \mu_H}[\sum_{i \neq j}^k|y_{i}|^2 |y_{j}|^2] - \frac{k^2}{N^2} \label{Var_Haar}
\end{eqnarray}

We know that the Haar measure and quantum circuits which are 2-designs (with $O(\log(N)$) depth anti concentrate (\cite{Hangleiter_2018,Dalzell_2022}). Specifically for Haar measure at large N, we have
\begin{equation}
   \sum_i^N E_{U\sim \mu_H}[|y_i|^4] = \frac{2}{N}
   \label{fourth}
\end{equation}
Using $(\sum_{i=1}^N(|y_{i}|^2))^2=1$, we get 
\begin{equation}
    E_{U\sim \mu_H}[\sum_{i \neq j}^N(|y_i|^2 |y_j|^2)] =\frac{N-2}{N}
\end{equation}

Then, similar to our previous arguments, we can show that 
\begin{equation}
    E_{U\sim \mu_H}[(|y_{i}|^4)] = E_{U\sim \mu_H}[(|y_{j}|^4)] 
    \label{fourthequal}
\end{equation}
$\forall$ $i,j$ and 
\begin{equation}
    E_{U\sim \mu_H}[(|y_{i}|^2|y_{j}|^2)] = E_{U\sim \mu_H}[(|y_{k}|^2|y_{l}|^2)] 
    \label{twotwoequal}
\end{equation}
 $\forall(i,j)$ and $k,l$. (Eqs. \ref{fourthequal},\ref{twotwoequal}) can be proved by using V in Eq.\ref{Haar} to be $i \leftrightarrow j$ basis switch and $(i,j) \leftrightarrow (k,l)$ basis switch. Eqs. \ref{fourthequal},\ref{twotwoequal} allows us to write Eq.\ref{Var_Haar} as,\\
\begin{equation}
    Var_{U \sim \mu_H} = \frac{2k}{N^2} + \frac{N-2}{N}\frac{k^2-k}{N^2-N} - \frac{k^2}{N^2} = \frac{(N-k)k}{N^2(N-1)}
\end{equation}
The variance for the $|\Pi \cdot \vec{v}|^2_2$ would be $\frac{N-k}{k(N-1)}$ which scales as $O(\frac{N-k}{Nk})$ for large $N,k$. Using the derived variance and the deviation from the mean to be $2 \epsilon$  gives (using Chebysev's inequality)
\begin{equation}
    P_{U\sim \mu_H}[||\Pi \cdot \vec{v}|_2-1| \leq \epsilon] \geq (1- \frac{N-k}{4 kN\epsilon^2})
\end{equation}

\end{proof}

\begin{theorem}(Theorem \ref{2designrp} in main manuscript)
Let $ U \in$ $\mathbb{U}(N)$ be sampled uniformly from the $\alpha$ approximate unitary 2- design ($\mu_{2,\alpha}$) and let $\vec{x_1}$, $\vec{x_2}$ $\in$ $\mathbb{R}^N$. Then the matrix $\Pi^{k\times N}$ obtained by considering any k rows of $U$ followed by multiplication with $\sqrt{\frac{N}{k}}$ satisfy
\begin{equation}
     (1-\epsilon)|\vec{x_1}-\vec{x_2}|_2\leq|\Pi \cdot (\vec{x_1} -\vec{x_2})|_2\leq (1+\epsilon)|\vec{x_1}-\vec{x_2}|_2
\end{equation}
with a probability greater than $1- (\frac{N-k}{4 kN\epsilon^2} + \frac{\alpha}{4\epsilon^2 k})$ with $\epsilon \in R_{\ge 0}$ being the error threshold.
\end{theorem}
\label{Proof2}

\begin{proof}

Define a function 
\begin{equation}
    M(U) = (\ket{v}\bra{v} - \frac{1}{N})^{\otimes 2}
\end{equation}
Let $Var_{\mu_{2,\alpha}}$ be defined as follows
\begin{eqnarray}
Var_{\mu_{2,\alpha}} =E_{U\sim \mu_{2,\alpha}}[(\sum_{i=1}^k|y_{i}|^2)^2] - E^2_{U\sim \mu_{2,\alpha}}(\sum_{i=1}^k|y_{i}|^2) \\
= E_{U\sim \mu_{2,\alpha}}[\sum_{i=1}^k|y_{i}|^4] + E_{U\sim \mu_{2,\alpha}}[\sum_{i \neq j}^k|y_{i}|^2 |y_{j}|^2] - \frac{k^2}{N^2} \label{Var_2design}
\end{eqnarray}

In terms of $M(U)$ we can write Eq. \ref{Var_Haar} and Eq.\ref{Var_2design} as 
\begin{equation}
    Tr(E_{U \sim \mu_{2,\alpha}}[P_k^{\otimes 2}U^{\otimes 2}(M(U))(U^\dagger)^{\otimes 2}P_k^{\otimes 2}]) = Var_{\mu_{2,\alpha}}
\end{equation}
\begin{equation}
    Tr(E_{U \sim \mu_H}[P_k^{\otimes 2}U^{\otimes 2}(M(U))(U^\dagger)^{\otimes 2}P_k^{\otimes 2}]) = Var_{\mu_H}
\end{equation}
Then,
\begin{equation}
    |Var_{\mu_{2,\alpha}}-Var_{\mu_H}| = 
    |Tr(E_{U \sim \mu_{2,\alpha}}[P_k^{\otimes 2}U^{\otimes 2}(M(U))(U^\dagger)^{\otimes 2}P_k^{\otimes 2}]
      - E_{U \sim \mu_H}[P_k^{\otimes 2}U^{\otimes 2}(M(U))(U^\dagger)^{\otimes 2}P_k^{\otimes 2}])|
\end{equation}
\begin{equation*}
        \leq |P_k^{\otimes 2}|_2|E_{U \sim \mu_{2,\alpha}}[U^{\otimes 2}(M(U))(U^\dagger)^{\otimes 2}]-E_{U \sim \mu_H}[U^{\otimes 2}(M(U))(U^\dagger)^{\otimes 2}]|
\end{equation*}
\begin{equation}
    \leq k \frac{\alpha}{N^2}
    \label{v_bound}
\end{equation}
    where we used the monomial definition of approximate unitary 2-designs and its equivalence (See \cite{low2010pseudorandomness}) to the diamond norm definition used in the main text.
Using the theorem \ref{Haartheorem}, we get 
\begin{equation}
    Var_{\mu_{2,\alpha}} \geq \frac{(N-k)k}{N^2(N-1)} - k \frac{\alpha}{N^2} 
\end{equation}
The variance for the $|\Pi \cdot \vec{v}|^2_2$ would be greater than $\frac{N-k}{k(N-1)} - \frac{\alpha}{k} $ which scales as $O(\frac{N-k}{Nk}) - \frac{\alpha}{k} $ for large $N,k$. Using the derived variance (upper bound from Eq.\ref{v_bound}) and the deviation from the mean to be $2 \epsilon$   gives (using Chebysev's inequality )
\begin{equation}
    P_{U\sim \mu_H}[||\Pi \cdot \vec{v}|_2-1| \leq \epsilon] \geq 1- (\frac{N-k}{4 kN\epsilon^2} + \frac{\alpha}{4\epsilon^2 k})
\end{equation}
\end{proof}

\begin{theorem}
(Theorem \ref{entropy accuracy} in main manuscript)
For a random matrix $\Pi$ satisfying the Johnson-Lindenstrauss (JL) lemma with a distortion $\epsilon \sqrt{\delta}$, where $\epsilon \leq 1/6$ and $\delta \leq 1/2$, the difference in the Von Neumann entropy of a density matrix $\rho$ computed using the random projection with $\Pi$ (denoted as $\tilde{S}(\rho)$) and the true entropy ($S(\rho)$) can be bounded as follows
\begin{equation}
    |{\tilde{S}(\rho) - S(\rho)}| \leq \sqrt{3\epsilon}S(\rho) + \sqrt{\frac{9}{2}}\epsilon
\end{equation}
with probability at least ($1-\delta$).
\end{theorem}
\label{Proof3}
 To prove the theorem we will need to use theorems 10 and 13 of \cite{kontopoulou2020randomized}. According to theorem 13, let $\textbf{DAD}$ be a symmetric positive definite matrix such that $\textbf{D}$ is a diagonal matrix and $A_{ii} = 1$ for all $i$. Also, let $\textbf{DED}$ be a perturbation matrix such that $|E|_2 \leq \lambda _{min} (A)$. Now, let $\lambda _i$ be the $i$-th eigenvalue of $\textbf{DAD}$ and let $\lambda_i^{'}$ be the eigenvalue of $\textbf{D(A+E)D}$.  Then, $ \forall i,$
\begin{equation}
    |\lambda_i - \lambda_i^{'}| \leq \frac{|E|_2}{\lambda_{min}(A)}
    \label{eig perturbation}
\end{equation}
 Now, consider a distribution $\mathcal{D}$ on matrices $\Pi \in \mathbb{R}^{k\times N}$ (or $\mathbb{R}^{N \times k}$ if data vectors get multiplied from the right) satisfying 
\begin{equation}
    E_{\Pi \sim \mathcal{D}}||\Pi x|^2_2 -1|^2 \leq \epsilon^2 \delta
    \label{JL moment}
\end{equation}
then, according to theorem 13 of \cite{Woodruff_2014}, for any orthonormal matrix $O$ with $N$ rows 
\begin{equation}
    Pr_{\Pi \sim D}[|O^T \Pi^T \Pi O - I|_2 \geq 3\epsilon] \leq \delta
    \label{2nd order JL}
\end{equation}\\

Having these two results, we can derive the bounds on the accuracy of entropy computed post-random projection. In this regard, let us look at a general semi-positive definite density matrix $\rho$ which could be written as $\rho = W \Sigma_p W^T$ where $W$ has orthonormal columns and $\Sigma_p$ is a diagonal matrix containing the singular values of $\rho$. We note that the eigenvalues of $\rho \rho^T = W \Sigma^2_p W^T$ are equal to the eigenvalues of the diagonal matrix $\Sigma_p^2$. Similarly, the eigenvalues of $W\Sigma_p W^T \Pi \Pi^T W \Sigma_p W^T$ are equal to the eigenvalues of $\Sigma_p W^T \Pi \Pi^T W \Sigma_p$. Now, using the Eq.\ref{eig perturbation}, we can compare the eigenvalues of the matrices
\begin{equation*}
    \Sigma_p I_k \Sigma_p \hspace{5 pt} \rm{and} \hspace{5 pt} \Sigma_p W^T \Pi \Pi^T W \Sigma_p
\end{equation*}
Using $E = W^T \Pi \Pi^T W - I_k$ in the Eq. \ref{2nd order JL}, we get that
\begin{equation}
    |E|_2 \leq 3\epsilon \le 1
    \label{E limit}
\end{equation}
with a probability greater than $1-\delta$.\\\\
The eigenvalues of $\Sigma_p I_k \Sigma_p$ are equal to $p_i^2$ for $i=1,..,k$ and the eigenvalues of the $\Sigma_p W^T \Pi \Pi^T W \Sigma_p$ are equal to $\tilde{p_i}^2$, where $\tilde{p_i}$ are the singular values of $\Sigma_p W^T \Pi$. These are exactly equal to the singular values of $W \Sigma_p W^T \Pi = \rho \Pi$. This along with Eqs.\ref{E limit},\ref{eig perturbation} leads us to conclude:
\begin{equation}
    |p_i^2 - \tilde{p_i}^2| \leq 3\epsilon p_i^2
\end{equation}
This result guarantees that the singular values of $\rho$ (from the main text) are captured with a perturbative factor of $3\epsilon$. Using this, we can find bounds to the error in the entropy computed after the random projection. We start with the upper bound
\begin{eqnarray*}
    \Sigma_{i=1}^k \tilde{p}_i \ln (\frac{1}{\tilde{p}_i}) \leq \Sigma_{i=1}^k (1+3\epsilon)^{1/2}p_i \ln (\frac{1}{(1-3\epsilon)^{1/2}p_i}) \\
    \leq (1+3\epsilon)^{1/2} S(\rho) + \frac{\sqrt{1+3\epsilon}}{2} \ln (\frac{1}{1-3\epsilon}) \\
    \leq (1 + 3\epsilon)^{1/2} S(\rho) + \frac{\sqrt{1+3\epsilon}}{2} \ln (1+6\epsilon)\\
    \leq (1+\sqrt{3\epsilon}) S(\rho) + \sqrt{\frac{9}{2}} \epsilon
\end{eqnarray*}
In the second to last inequality, we used $1/(1-3\epsilon) \leq (1+6\epsilon)$ for any $\epsilon \leq 1/6$, and in the last inequality we used $\ln (1+6\epsilon) \leq 6\epsilon$ for $0 \leq \epsilon \leq 1/6$. Similarly, we can find the lower bound to be
\begin{equation*}
    \Sigma_{i=1}^k \tilde{p}_i \ln (\frac{1}{\tilde{p}_i}) \geq (1 -\sqrt{3\epsilon}) S(\rho) - {\frac{3}{2}}\epsilon
\end{equation*}
Combining both that bounds we get the bound for the error in entropy obtained after random projection $\tilde{S}$ wrt true entropy $S$
\begin{equation}
    |\tilde{S}(\rho)-S(\rho)| \leq \sqrt{3\epsilon}S(\rho) + \sqrt{\frac{9}{2}}\epsilon 
\end{equation}

\def\mathunderline#1#2{\color{#1}\underline{{\color{black}#2}}\color{black}}
{\color{black} \section{Image reconstruction post-quantum random projection}
\label{imagereconsapp}
The figure \ref{fig: recons} was a demonstration of the reconstruction of one image vector of the dataset. Since the dimensionality reduction is for the dataset as a whole, here we attach the plot of average $|\vec{x}-\vec{x}_{recons}|$ for all the image vectors $\vec{x}$ (the subset of MNIST containing 1000 images) and their reconstructed vectors $\vec{x}_{recons}$ in Fig.\ref{Fig:quant_image_recons} for various reduced dimensions along with their $95 \% $ confidence intervals.

\begin{figure}[h!]
    \centering   \includegraphics[width=0.45\textwidth]{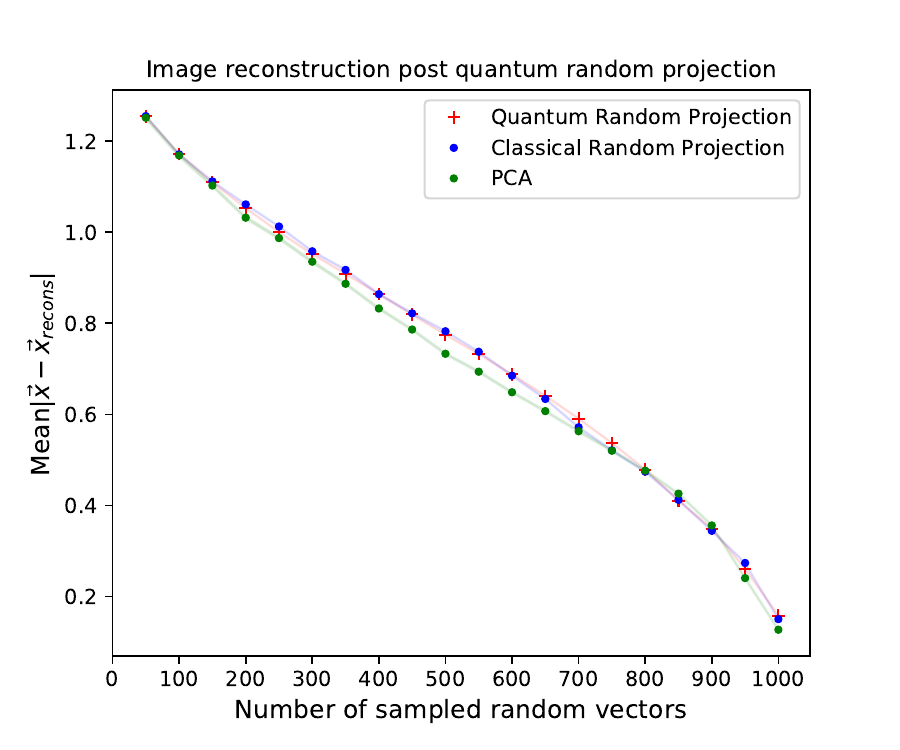}
    
    \label{Fig:quant_image_recons}
    \caption{The average norms of difference between image data vectors and the image vectors reconstructed from PCA, Classical Random Projection and Quantum Random Projection schemes for varying reduced dimensions.}
\end{figure}
The average norm of the difference between any two unit norm vectors in 1024 dimensional space is $\sqrt{2} \sim 1.414$. This sets a standard to compare the average norms plotted in the \ref{Fig:quant_image_recons}. \\\\
As mentioned in the experiments section, we assume that we can make arbitrary projection operators with any number of basis vectors, i.e, if the $P_k$ operator projects to the first $k$ basis state $(\ket{p_1},\ket{p_2},..,\ket{p_k})$ in any basis. Then, 
\begin{equation}
    P_k = \sum_{i=1} ^k \ket{p_i}\bra{p_i}
\end{equation}
where we don't have any restriction on what basis we pick and what values k can take. Quantum projectors are easy to construct when k takes values 128,256 and 512 (corresponding to measuring 1,2 and 3 qubits respectively). Though it is impractical to construct such quantum projection operators for $k=400,700$, we used those values only to compare the performance of quantum random projection against classical random projection which doesn't have any restrictions on the values $k$ can take. \\\\
For the specific numbers mentioned in the figure, we used the first 400 and 700 components of the wave vector after the random quantum circuit ($U$) and added extra basis states with zero amplitudes to make them 1024 dimensional. Then, the wavevectors are reconstructed by the action of the inverse of $U$ i.e., $U^\dagger{}$. We understand that this process is more insightful when the projection operators are just 1,2 or 3 qubit measurements to specific states ($\ket{0}$ or $\ket{1}$). There, if the quantum random projection is made such that one of the qubits is projected to say $\ket{0}$ state. Then, reconstruction is done by appending the $U^\dagger{}$ circuit to the reduced quantum state + measured qubit in $\ket{0}$ (This is equivalent to adding zeros to the basis elements where the measured qubit is in $\ket{1}$ state just like how we boosted dimensions from k=400 or 700 to 1024.)\\}
\section{Performance of quantum random projection on exponential decay profile}
\label{EXP}
\begin{figure}[htbp]
  \centering
  \includegraphics[width=0.8\textwidth]{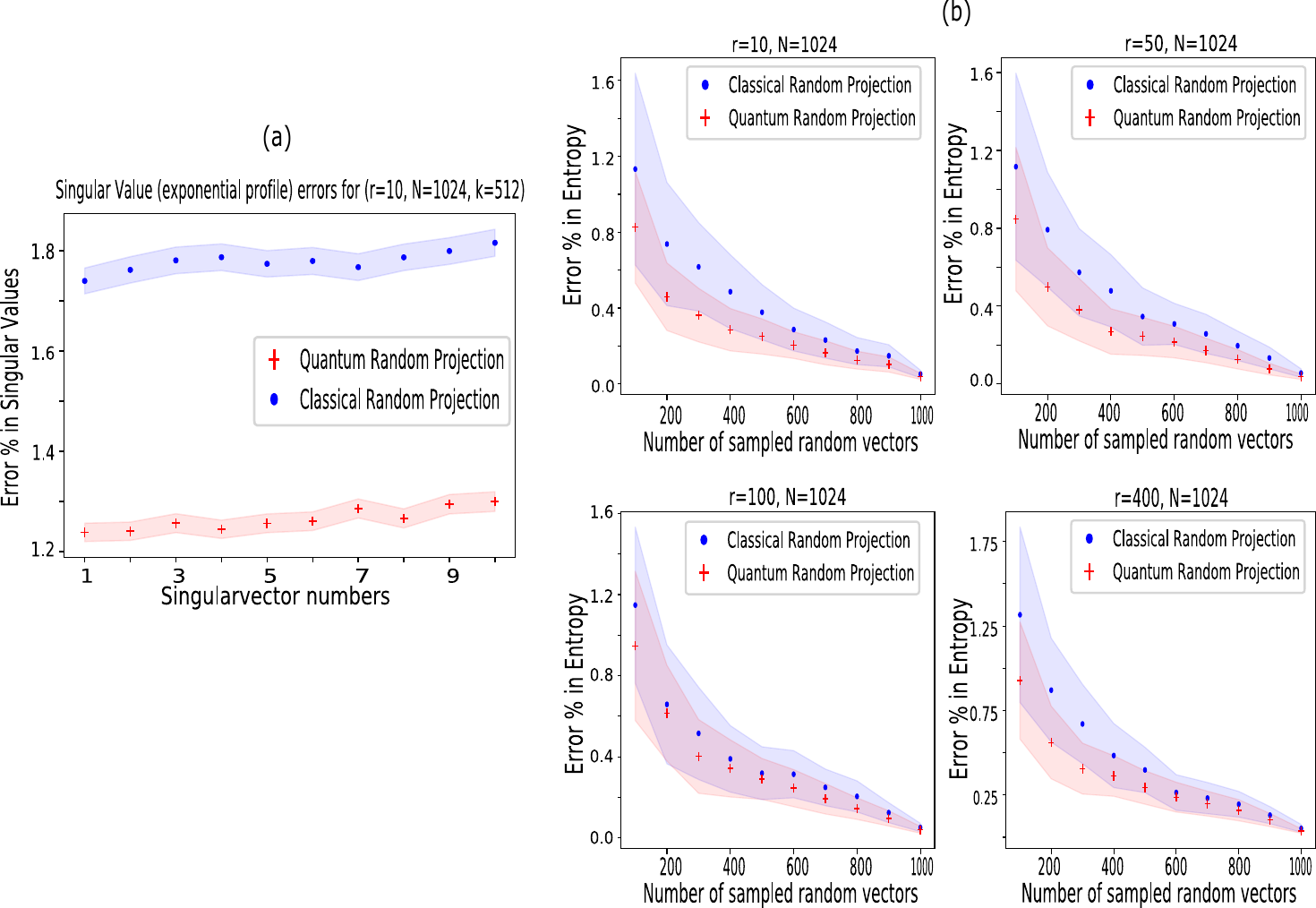}
  \caption{The figure contains the analogous plots discussed in the main text for the density matrices with exponential decay profile. (a) shows the accuracy with which the quantum random projector and the classical random projector pick the singular values of the density matrix for $r = 10$ when reducing the system size by half. The envelope represents 95 \% confidence intervals by running the experiments over 10,000 randomly generated density matrices. The error \% observed for exponential decay profile is similar to the values we obtained for the linear decay profile. (b) shows the accuracies of quantum random projection and classical random projection in the entropy computation of randomly generated density matrices of size N=1024 and ranks $r=10,50,100,400$ with exponentially decaying profile. The envelopes represent their 90 \% confidence intervals by running the experiments over 100 randomly generated density matrices. The accuracies remain constant because the singular values  don't vary much while varying the ranks owing to the rapid drop in successive singular values in an exponential decay profile.}
  \label{fig: exp}
\end{figure}

\section{Performance of quantum random projection on larger datasets}
\label{LARGER}
\begin{figure}[htbp]
  \centering
  \includegraphics[width=0.8\textwidth]{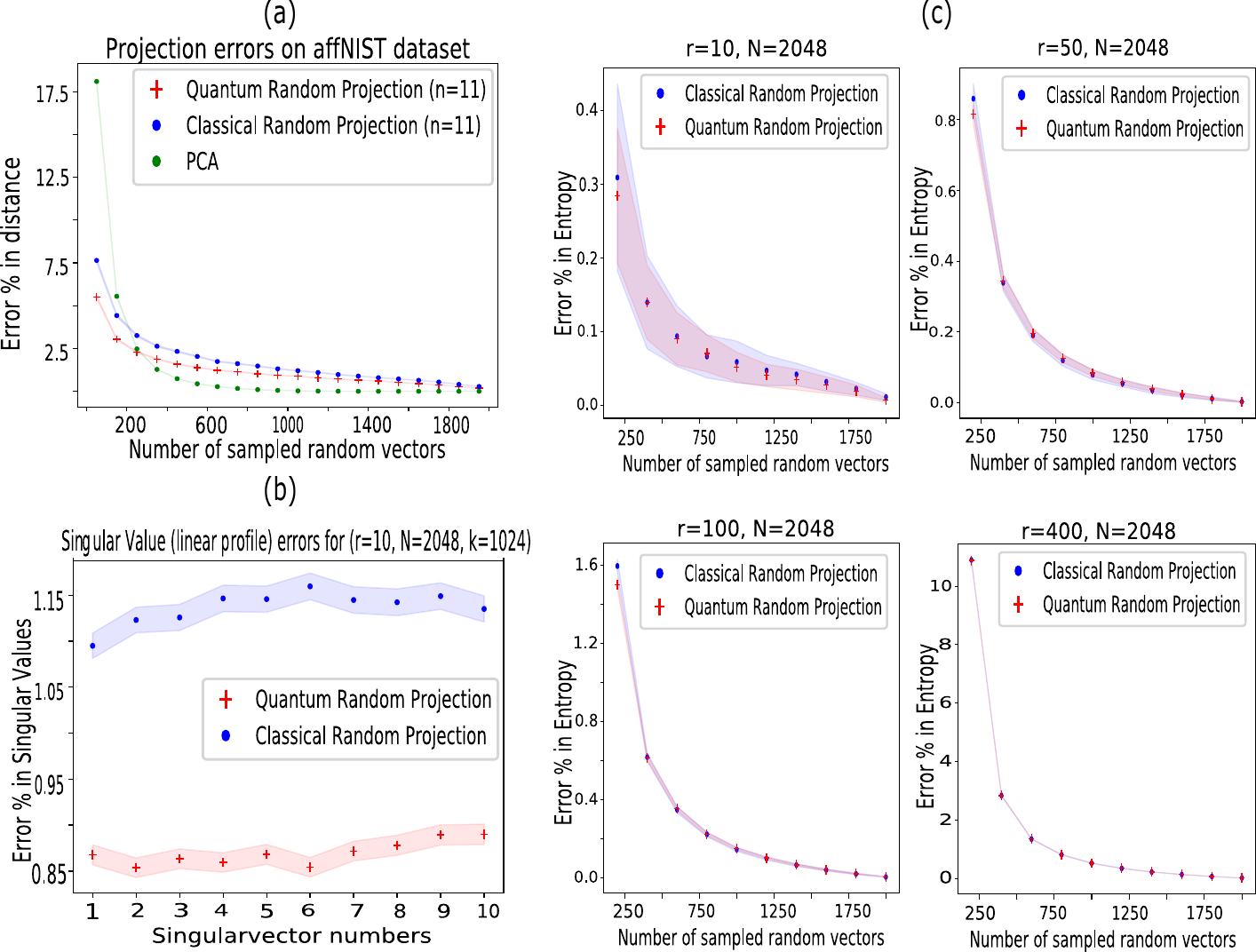}
  \caption{The figure contains the analogous plots discussed in the main text for system size $= 2048$ ($n=11$). The error \% observed for the $n=11$ is better than the values obtained for the $n=10$ case . This is because of improved efficiency in random projection for larger datasets as the error bounds derived in the theorems in our main text become tighter.  (a) shows the mean percentage errors in the distance between 10,000 different random pair of data vectors in the affNIST image dataset which are 40 $\times$ 40 images and was boosted to 2048 $\times$ 1 datavectors for our experiment. Here, we can clearly notice that at lower reduced dimensions, the random projection methods perform better than even PCA. Everywhere , the quantum random projection performs slightly better than the classical random projection. (b) shows the accuracy with which the quantum random projector and the classical random projector pick the singular values of the density matrix for $r = 10$ when reducing the system size by half. The envelope represents 95 \% confidence intervals by running the experiments over 10,000 randomly generated density matrices. (c) shows the accuracies of quantum random projection and classical random projection in the entropy computation of randomly generated density matrices of size N=2048 and ranks $r=10,50,100,400$  with linearly decaying singular value profile. The envelopes represent their 90 \% confidence intervals by running the experiments over 100 randomly generated density matrices. The accuracies improve with a decrease in the rank of the system just  like the $n=10$ case. However, the accuracies here are better than $n=10$ case due to the improved efficiency of random projection for larger datasets.}
  \label{fig: 11}
\end{figure}
\newpage

\section{Performance of quantum random projection constructed using lesser expressive local random quantum circuits}
\label{NON EXACT}
\begin{figure}[htbp]
  \centering
  \includegraphics[width=\textwidth]{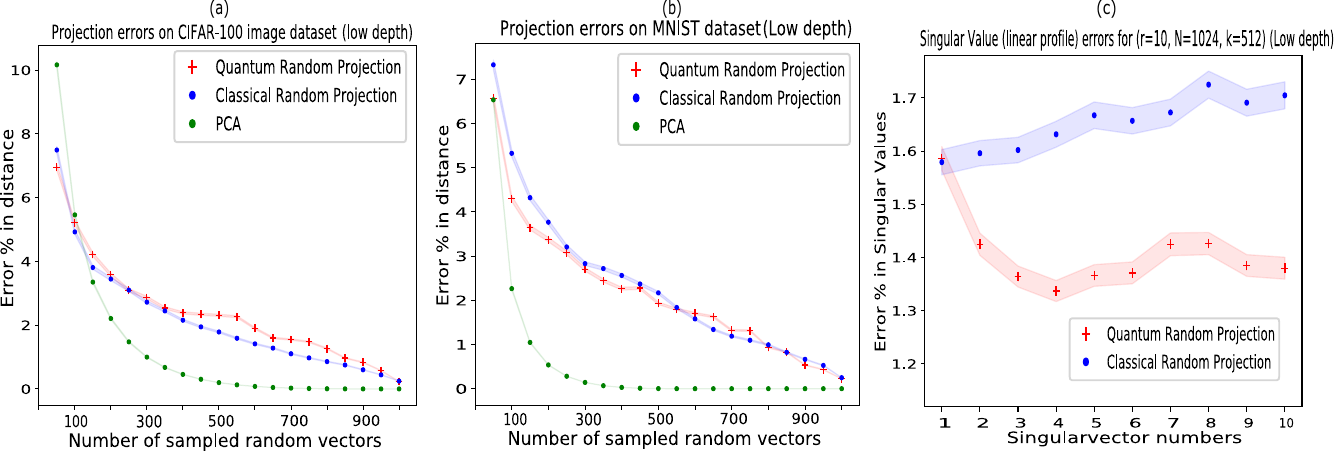}
  \caption{The figure contains the analogous plots discussed in the main text but the quantum random projector used here has been obtained from a local random quantum circuit far from the exact 2-design limit. This is our proxy for approximate 2-design. More specifically, we picked a depth of 50 in the quantum circuit shown in the Fig.\ref{fig:rqc} . Figures (a) and (b)  show the mean percentage errors in the distance between 10,000 different random pair of data vectors in the MNIST image dataset and CIFAR-100 image dataset which are 28 $\times$ 28 images and was boosted to 1024 $\times$ 1 data vectors for our experiment. The performance of this quantum random projector is not as good as the exact 2-design quantum random projector.  This is because the error bounds derived in the theorems for approximate unitary 2-designs are not as tight as an exact 2-design quantum random projector. Fig.(c)  shows the accuracy with which the approximate unitary 2-design quantum random projector and the classical random projector pick the singular values of the density matrix for $r = 10$ when reducing the system size by half. The envelope represents 95 \% confidence intervals by running the experiments over 10,000 randomly generated density matrices. The drop in performance of the quantum random projector here could be attributed to the error bounds in JL lemma not being as tight as the exact 2-design quantum random projector. }
  \label{fig: low}
\end{figure}
\newpage
\section{Details regarding the VQSVD performed to retrieve the dominant singular vectors}
\label{VQSVD App} 
The variational quantum singular value decomposition was used to retrieve the dominant singular vectors of a randomly generated data matrix with rank ($r=5$) and singular values following a linearly decaying profile. { \color{black} The matrix on which one has to perform quantum random projection and quantum SVD needs to be loaded as a sum of unitaries or Pauli strings (unit depth). This process has exponential complexity for an exact representation of the matrix. But with importance sampling (as mentioned in the \cite{Wang2021variationalquantum}) one only uses a subset of Pauli strings which approximates the matrix to sufficient accuracy. However, for our computation, just like the exact encoding scheme, we used the exact matrix. The demonstration of accurately retrieving the Singular vectors implies the same when the importance sampling creates the matrix with high accuracy.}\\

Since, the system size we used for this variational algorithm is large, we had to use the block initialisation strategy discussed in Ref.\cite{Grant2019initialization} with two identity blocks in our training ansatz to avoid the barren plateau issue \cite{McClean_2018}. Each block used in our ansatz has a hardware efficient ansatz circuit \cite{Kandala_2017} of depth 25 followed by it's inverse circuit (this part will remain an inverse circuit only at the beginning of the training. During training, the parameters of these two parts update independently) to start the training procedure close to identity and hence avoid the barren plateau. The accuracies with which we retrieved the singular values and the dominant singular vectors could also be improved with other ansatze which avoid barren plateau.\\

{\color{black} The projection operators in the figure are just measurements of certain qubits and making sure they are in a certain state ($\ket{0}$) or operators of the form $\frac{1}{2}(1+\sigma_z^i)$ which projects to $\ket{0}$ state of qubit $i$.}
\begin{figure}[htbp]
  \centering
  \includegraphics[width=\textwidth]{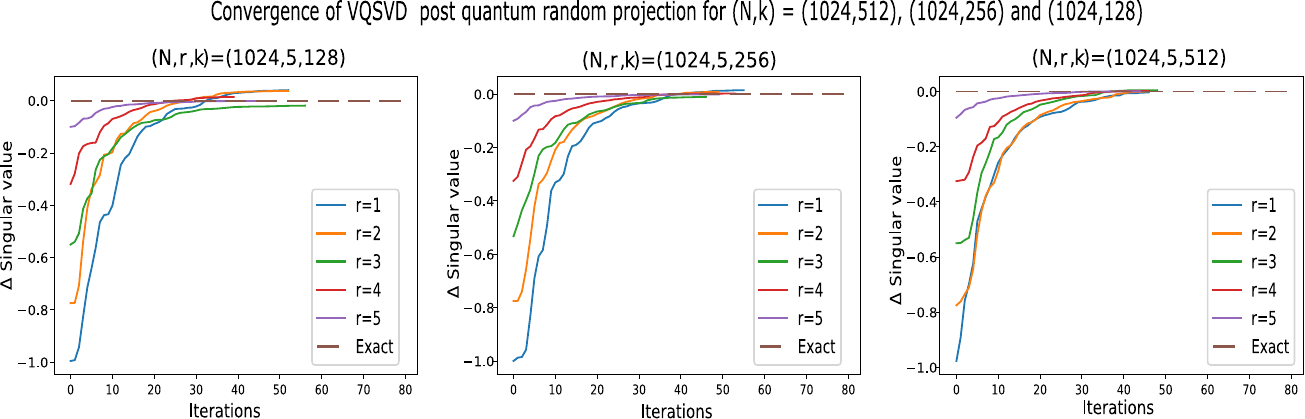}
  \caption{The convergence of singular values obtained using the variational quantum singular value decompositon algorithm to their true values post quantum random projection (2-design quantum random projector). We observe that some times the converged singular value is greater than the true value, this is beacuse of the distortion of singular values created by the quantum random projection. }
  \label{fig: VQSVD convergence}
\end{figure}

\end{document}

%% file: main_manuscript.bbl
%